\documentclass[AER,reviewmode]{Templates/AEA}

\usepackage[margin=1in]{geometry}
\newtheorem{theorem}{Theorem}

\newtheorem{lemma}{Lemma}

\newtheorem{proposition}{Proposition}

\usepackage[authoryear]{natbib}
\usepackage{amssymb}
\usepackage{graphicx}
\usepackage{tabularx}

\input{Templates/tcilatex}

\draftSpacing{1.5}

\begin{document}

\date{%
\today%
}
\title{Two Margins of Climate Cooperation: Emissions and Coalition Membership under Tipping}

\shortTitle{ }

\author{Yongyang Cai, Yaofeng Cui, Hongbo Duan, and Lei Zhu \thanks{Cai (Corresponding Author): Department of Agricultural, Environmental, and Development Economics, The Ohio State University, OH 43210, USA (cai.619@osu.edu); Cui: School of Economics and Management, Beihang University, 100191, Beijing, P. R. China (cuf01@buaa.edu.cn); Duan: School of Economics and Management, University of Chinese Academy of Sciences, 100190, Beijing, P. R. China (hbduan@ucas.ac.cn); Zhu: School of National Safety and Emergency Management, Beijing Normal University, 100875, Beijing, P. R. China (leizhu@bnu.edu.cn). An earlier version is "Long Coalition Leads to Shrink? The Roles of Tipping and Technology-Sharing in Climate Clubs." Acknowledgements: We thank Prof. Xiaobing Zhang, Dr. Xin Liu, Dr. Xinyuan Ren, and Dr. Zhihao Yan for their contributions on the algorithm design in the early stages of this research. We are grateful for comments from the participants at 2026 Africa Meeting of the Econometric Society, 2026 CEEP-BIT \& IIASA Joint Workshop on Energy and Climate Change,  PACE 2026 China International Symposium on Green and Low-Carbon Development, Workshop of Climate Finance \& Risk 2025, and seminars at Indiana University, University of Chinese Academy of Sciences, and Wuhan University. Duan gratefully acknowledges financial support from the National Natural Science Foundation of China through Grant 72325008; and Zhu gratefully acknowledges financial support from the National Natural Science Foundation of China through Grants 72342002, 72122002. The order of authors is alphabetical, and all authors contributed equally to this work.}}

\pubMonth{ }
\pubYear{ }
\pubVolume{ }
\pubIssue{ }


\begin{abstract}

We study climate coalition stability with endogenous temperature and stochastic tipping. In our dynamic game, rising temperature widens the emissions gap between members and nonmembers, shrinking stable coalitions over time. The threat of tipping strengthens free-riding, raising emissions and reducing coalition size; its occurrence also contracts coalitions numerically. Advanced regions form the core, while emerging \& developing regions drive membership variation. Comparing policy instruments, connection (internal incentives) and sanctions (external penalties) sustain similar coalitions, but connection raises aggregate payoffs, whereas sanctions achieve larger emissions cuts. As tipping approaches, this trade‑off intensifies, making the optimal institutional design state‑dependent.

\end{abstract}

\maketitle









The Paris Agreement provides the near-universal institutional foundation for contemporary international climate cooperation. It combines a common temperature objective with nationally determined contributions, allowing countries to tailor their climate commitments to domestic conditions. This flexibility enables broad participation but limits enforcement, as the global benefits and local costs of mitigation create free‑riding incentives that weaken the agreement \citep{nordhaus2019climate,farrokhi2025can}. The central challenge after the Paris Agreement is therefore not only to establish collective climate goals, but also to sustain participation and mitigation effort as countries' incentives evolve.

Climate coalitions offer a prominent response to this enforcement problem \citep{paroussos2019climate,chan2022assessing}. The Powering Past Coal Alliance, for example, coordinates commitments to phase out unabated coal power, while proposed sectoral clubs seek to decarbonize emissions-intensive industries such as steel \citep{jewell2019prospects,hermwille2022climate}. By enabling a subset of jurisdictions to coordinate more ambitious action, these arrangements can deepen cooperation without requiring universal participation. They do not, however, eliminate free-riding. Coalition members bear the costs of coordinated mitigation, while the resulting climate benefits also accrue to nonmembers \citep{nordhaus2015climate,battaglini2016participation}. The same payoff asymmetry that gives rise to free-riding also makes coalition stability sensitive to the climate state. Rising temperature therefore reshapes both emissions within a coalition and the composition of it. The question is thus not simply whether a coalition is stable at a given point in time, but how its stable membership evolves with the climate state.

Coalition stability is well characterized in the literature, yet how it varies with the climate state remains unclear. Models of international environmental agreements show that free-riding, regional heterogeneity, and surplus sharing jointly determine the size and composition of stable coalitions \citep{carraro1993strategies,barrett1994self,weikard2006impact,mcginty2007international}. Dynamic models further assess whether self‑enforcement, renegotiation, and political incentives can sustain participation and compliance, and how these forces influence agreement duration, investment, and technology adoption \citep{dutta2004self,martimort2016mechanism,harstad2012climate,harstad2016dynamics,harstad2019compliance,battaglini2020political}. Closest to ours, \citet{wagner2016estimating} models treaty formation as an infinite-horizon participation game, but represents changes in participation benefits through a smooth time trend. Our departure is to make temperature endogenous to regional emissions and characterize coalition stability along the resulting warming path. Solving the model numerically reveals that the stable coalition shrinks as temperature rises: advanced regions form its core, while variation in membership is concentrated among emerging \& developing regions. 

The shifting composition of climate coalitions makes institutional design critical for stability by altering members’ participation incentives. We focus on two broad institutional channels---internal incentives and external penalties---through which participation can be sustained. Internal incentives increase the payoff from coalition membership. In our model, these incentives arise from \textit{connection}, which enables coalition members to benefit from technology cooperation, investment complementarity, tariff reduction, international transit connectivity, and technological spillovers \citep{barrett2006climate,helm2015climate,rubio2017sharing,eichner2022environmental}. External penalties instead reduce the payoff of remaining outside regions from nonparticipation. The trade-agreement literature highlights dispute settlement, contractual incompletness, multilateral institutions and reciprocity as key mechanisms for sustaining cooperation  \citep{maggi1999role,bagwell1999economic,horn2010trade,maggi2011role}. In climate settings, carbon tariffs and border measures similarly raise the cost of staying out, although their effectiveness depends on retaliation, membership, and policy coordination  \citep{bohringer2016strategic,al2020measures,hsiao2026coordination,nordhaus2021dynamic}. Within a common dynamic framework, we compare connection and sanctions and find a clear ranking: both instruments sustain coalitions of similar size, but connection raises aggregate payoffs, whereas sanctions achieve greater emissions reductions.

A further complication is that climate cooperation unfolds under nonlinear climate dynamics. Standard climate-economic models have difficulty representing deep uncertainty and low-probability catastrophic damages \citep{weitzman2009modeling,pindyck2013climate}, even as abrupt shifts and critical thresholds are well established for major Earth-system components under anthropogenic forcing \citep{alley2003abrupt,lenton2008tipping,kriegler2009imprecise}. Economic studies show that tipping can raise optimal carbon prices and create a precautionary motive for mitigation, particularly when multiple tipping elements interact \citep{lemoine2014watch,lontzek2015stochastic,cai2015environmental,lemoine2016economics,cai2016risk,ploeg2018}. Related work examines uncertainty, learning, and the implications of tipping for mitigation and the social cost of carbon \citep{berger2017managing,cai2017,cai2019social,dietz2021economic,bremer2021,hong2023}. Yet none of these models consider the coalition dynamics that we focus on in this paper. We distinguish between the threat of tipping (where the tipping probability becomes nonzero) and its occurrence, and examine how each affects regional emissions and coalition stability. Our results show that the threat of tipping induces precautionary emission cuts while simultaneously strengthening free‑riding incentives. Its occurrence, by contrast, has an ambiguous effect on emissions and tends to further contract the stable coalition.

We formalize these interactions in a dynamic stochastic game among heterogeneous regions. The framework advances the analysis of international environmental agreements along three dimensions. It jointly endogenizes regional emission choices and coalition participation, allowing changes in membership to affect emissions and, through temperature, future participation incentives. It also incorporates temperature-dependent stochastic tipping into coalition formation, so that both the threat and the occurrence of tipping can alter regional behavior. Finally, it combines analytical characterizations of emission choices and stability conditions with a numerical solution of the dynamic temperature and tipping process, making the state-dependent evolution of stable coalitions computationally tractable while preserving the model's strategic structure. The analysis yields main findings at four dimensions: 

\textit{Emissions and coalition participation as temperature rises.} Rising temperature affects international cooperation along two margins: emissions and participation. As temperature rises, the gap between coalition members' and nonmembers' optimal emissions widens, making remaining outside a coalition more attractive. Coalitions therefore shrink over time. Membership changes also induce sizable emission adjustments—leavers increase emissions, joiners reduce them. This contraction effect amplifies the direct impact of warming on global emissions.

\textit{Tipping threat and tipping occurrence.} The threat of tipping creates a precautionary incentive for mitigation but also strengthens free-riding incentives and reduces participation. Numerically, the latter effect dominates in equilibrium, raising aggregate emissions and shrinking the stable coalition. The occurrence of tipping induces discrete adjustments in emissions and membership, with an ambiguous net effect that depends on the climate state and the induced change in coalition composition. However, our numerical results show that, both the threat and the occurrence of tipping reduce coalition size and exacerbate free-riding.

\textit{Regional heterogeneity in coalition dynamics.} Advanced regions remain the core of the stable coalition, while its expansion and contraction reflect the entry and exit of emerging \& developing regions. The resulting coalition includes both groups and consistently dominates noncooperation on both economic and climate dimensions. The stable coalition yields lower aggregate mitigation than the grand coalition that consists of all regions, but it avoids the post‑tipping collapse in economic benefits. Its composition thus reflects an endogenous compromise between the interests among advanced and emerging \& developing regions.

\textit{Connection versus sanctions.} Even without tipping, connection and sanctions, though similar in lowering emissions and expanding coalitions, diverge in their effects on abatement costs and payoffs. Tipping widens these divergences: sanctions induce sharper emission reductions as the temperature approaches the tipping threshold, but at greater loss in payoffs, whereas connection delivers higher aggregate payoffs and encourage participation. Sustaining the grand coalition requires a higher sanction rate than connection. Thus, internal incentives are more effective for participation and payoffs, whereas external penalties achieve greater emissions reductions.

The remainder of the paper is organized as follows. Section \ref{Sec.Model} presents the model, numerical methods and theoretical results; Section \ref{Sec.Results} describes the parametrization and discusses results; Section \ref{Sec.Conclusions} concludes.

\section{Model}
\label{Sec.Model}
\subsection{Principle Assumptions}
\label{subsec:assumptions}
For long-term dynamics of climate coalition, we adopt a continuous-time model with $n$ heterogeneous regions. Regions are divided into coalition members ($i$) and non-members ($j$), and the coalition among member regions is expressed as $\mathcal{S}$ when there are $m$ members within it. The model incorporates the following assumptions: 

\textbf{Emissions and Production Benefits}: Each region derives economic benefits from production, which is linked to its emissions level $q\left(t\right)$ at time $t$. The benefit function is assumed to be quadratic in emissions (\citealp{barrett2013climate,battaglini2016participation, li2019climate}): 
$$\pi \left(q(t)\right)=\alpha q\left(t\right)-\frac{1}{2}\beta q\left(t\right)^2$$

\noindent where $\alpha$ and $\beta$ are region-specific benefit parameters. For simplicity, we exclude benefits from non-emissions growth (e.g., clean energy development or decoupling of GDP from $\text{CO}_2$ emissions), as the primary focus of climate cooperation is to address the negative externalities associated with greenhouse gas emissions. Non-emissions growth is treated as exogenous and does not influence coalition dynamics. 

\textbf{Climate Damages}: Global average temperature anomaly $T\left(t\right)$ above preindustrial levels imposes two types of costs on regions: 

\textit{Continuous Losses}: These are increasing and convex in temperature, reflecting the escalating economic and ecological impacts of climate change. The continuous loss function is specified as (e.g., \citealp{mcginty2007international}; \citealp{nordhaus2014estimates}; \citealp{barnett2023climate}): 
$$D\left(T(t)\right)=\frac{1}{2}\rho T\left(t\right)^{2}+\eta T\left(t\right)$$

\noindent where $\rho$ and $\eta$ are region-specific loss parameters. 

\textit{Tipping Event Losses}: If temperature exceeds a critical threshold, regions incur an additional loss $\ell D(T(t))$, where $\ell$ is the damage rate. The loss function in the presence of a tipping event becomes:
$$(1+\ell)D\left(T(t)\right)$$ 

\textbf{Temperature-Emissions Relationship}: Temperature rise is modeled as a linear function of cumulative global emissions, following the transient climate response to emissions (TCRE) framework (\citealp{matthews2009proportionality}):
$$T\left(t\right)=\lambda\int_{0}^t {\left(\sum_{i\in \mathcal{S}} q_i\left(t\right)+\sum_{j \notin \mathcal{S}}q_j\left(t\right)\right)}dt+T_0$$

\noindent where $\lambda$ is the TCRE parameter, and $T_0$ is set as an initial temperature level. Accordingly, the evolution of temperature is governed by:  
$$\frac{dT\left(t\right)}{dt}=\lambda \left(\sum_{i \in \mathcal{S}} q_i\left(t\right)+\sum_{j \notin \mathcal{S}}q_j\left(t\right)\right)$$

\textbf{Hazard Rate of Tipping Events}: The likelihood of a climate tipping event is modeled with a hazard rate $H\left(T\right)$, which increases with temperature rise. Following \citet{cai2016risk} and \citet{cai2019social}, we specify: 
$$H\left(T\right)=\chi \max\left(0,T-\underline{T}\right)$$
\noindent where $\chi$ is the hazard rate parameter and $\underline{T}$ is the temperature threshold below which the hazard rate is zero. This hazard rate is incorporated into the continuous-time models' Hamilton-Jacobi-Bellman (HJB) equations to account for the impact of tipping on regions' value functions.

\textbf{Coalition Formation and Connection}: Member regions collaborate to maximize the collective net benefits of the coalition, while non-members optimize their individual net benefits. At any time $t$, regions decide whether to join, leave, or remain inside/outside the coalition. For $n$ regions, there are $2^{n}-n-1$ possible coalition structures, as a nonempty coalition must include at least two regions. For completeness, let $V(T;\mathcal{S})$ denote the value of coalition $\mathcal{S}$, and let $W_i^J(T;\mathcal{S})$ denote the value of region $i$ when it remains outside coalition $\mathcal{S}$. We define $V^J(T;\emptyset)=0$. A singleton $\{i\}$ is treated as noncooperation, so that $V^J(T;\{i\})=W_i^J(T;\emptyset)$.

To reflect the benefits of cooperation, we introduce a connection effect within the coalition, which refers to additional technology cooperation, investment complementarity, tariff reduction,  international transit
connectivity, and technological spillovers. For example, \citet{paroussos2019climate} show that coalition membership enhances production efficiency through technology diffusion and low-cost climate financing. This effect is captured by a coalition-specific connection parameter $\tau$, which determines the magnitude of the benefit enhancement experienced by coalition members:
$$\tau = \frac{\sum_{i \in \mathcal{S}} \pi_{i,0}}{\sum_{i \in \mathcal{S}}  \pi_{i,0}+\sum_{j \notin \mathcal{S}}  \pi_{j,0}}  \bar{\tau}$$
\noindent where $\bar{\tau}$ denotes the connection rate attained under the grand coalition comprising all regions, and $\pi_{i,0}$ denotes the baseline benefit of region $i$ in the reference year.

Consequently, regions decide their emissions to balance the benefits of emissions-associated production against the costs of climate-related damages. Let $t_1$ denote the time at which the tipping event occurs. The instantaneous utilities of a coalition member $i$ before and after the tipping event are given by:
\begin{equation}
\label{Eq.V_utility}
u_i\left(q_i(t),T(t)\right) = 
\begin{cases}
\left(1+\tau\right)\pi_i\left(q_i(t)\right)-D_i\left(T(t)\right), &t < t_1\\
\left(1+\tau\right)\pi_i\left(q_i(t)\right)-(1+\ell)D_i\left(T(t)\right), &t \geq t_1
\end{cases}
\end{equation}

\noindent respectively. For a non-member region $j$, the corresponding utilities are: 
\begin{equation}
\label{Eq.W_utility}
u_j\left(q_j(t),T(t)\right) = 
\begin{cases}
\pi_j\left(q_j(t)\right)-D_j\left(T(t)\right), &t < t_1\\
\pi_j\left(q_j(t)\right)-(1+\ell)D_j\left(T(t)\right), &t \geq t_1
\end{cases}
\end{equation}

\subsection{Decisions with Single Tipping\label{subsec:one tipping}}

In this subsection, we analyze the strategic decisions of regions in the presence of a single climate tipping event. At any time $t$, the climate system can be in one of three states:

\textit{State ``$notip$'': No tipping.} Let $V^{\mathrm{notip}}$ and $W_j^{\mathrm{notip}}$ denote the value functions of the coalition and non-member $j$, respectively. The superscript ``$\mathrm{notip}$'' refers to the state in which no tipping event is present.

\textit{State ``$tipped$'': The tipping event has occurred.} Let $V^{\mathrm{tipped}}$ and $W_j^{\mathrm{tipped}}$ denote the value functions of the coalition and non-member $j$, respectively. The superscript ``$\mathrm{tipped}$'' indicates that the tipping event has already occurred.

\textit{State ``$pretip$'': Before the tipping event occurs.} Let $V^{\mathrm{pretip}}$ and $W_j^{\mathrm{pretip}}$ denote the value functions of the coalition and non-member $j$, respectively. The superscript ``$\mathrm{pretip}$'' indicates that a tipping event is possible but has not yet occurred.

Given the high degree of similarity between the no-tipping and post-tipping regimes, our subsequent analysis will focus primarily on the latter. For a given coalition $\mathcal{S}$, the coalition jointly determines the emission levels of all its members, $\{q_i(t)\}_{i\in\mathcal{S}}$, while non-member region $j$ chooses its emissions $q_j\left(t\right)$ individually. We examine the optimal emissions strategies and value functions for regions in each state separately.

\subsubsection{The Tipping Has Occurred}

In State ``$tipped$'', where the tipping event has already occurred, member and non-member regions simultaneously solve their respective optimization problems to maximize their expected discounted payoffs. The value function of coalition $\mathcal{S}$ in the post-tipping regime is given by:

\begin{equation}
V^{\mathrm{tipped}}(T;\mathcal{S})=\max_{q_i(t), {i \in \mathcal{S}}}\ \sum_{i\in \mathcal{S}}\int_{0}^{\infty}e^{-rt}\left(\left(1+\tau\right)\pi_i\left(q_i(t)\right)-(1+\ell)D_i\left(T(t)\right)\right)dt
\end{equation} 

\noindent Similarly, the value function of each non-member region $j\notin\mathcal{S}$ is given by:
\begin{equation}
W^{\mathrm{tipped}}_{j}(T;\mathcal{S})=\max_{q_j(t)}\ \int_{0}^{\infty}e^{-rt}\left(\pi_j\left(q_j(t)\right)-(1+\ell)D_j\left(T(t)\right)\right)dt
\end{equation}

\noindent where $r$ denotes the discount rate. Note that these value functions depend on the coalition $\mathcal{S}$. 

To analyze the game, we adopt the Markov strategy, as the Markov perfect equilibrium ensures dynamic consistency and sub-game perfection (\citealp{maskin2001markov}). The following proposition characterizes the value functions and optimal emission strategies in State ``$notip$'' and State ``$tipped$'', establishing key analytical properties of the equilibrium solution.

\begin{proposition}
\label{prop1}
The value functions in the absence of tipping event $V^{\mathrm{notip}}(T;\mathcal{S})$ and $W^{\mathrm{notip}}_j(T;\mathcal{S})$ and in the post-tipping regime $V^{\mathrm{tipped}}(T;\mathcal{S})$ and $W^{\mathrm{tipped}}_j(T;\mathcal{S})$ are strictly concave and quadratic in temperature $T$. Furthermore, the optimal emission strategies $q^{\mathrm{tipped}*}_i$, $q^{\mathrm{tipped}*}_j$ and $q^{\mathrm{notip}*}_j$, $q^{\mathrm{notip}*}_i$ are linear in $T$.
\end{proposition}

\begin{proof}
    The proof of Proposition \ref{prop1} is provided in Appendix \ref{appA1}. 
\end{proof}

Appendix \ref{appA1} derives the coefficients of the quadratic value functions for member and non‑member regions, while Appendix \ref{appA2} establishes the existence and uniqueness of their solutions. Recent literature demonstrates that mathematical singularities in differential games can be resolved as ``jump points" through the application of discontinuous Markovian strategies \citep{jaakkola2026markov}. Nevertheless, these frameworks are developed for symmetric Markov-perfect equilibria (\citealp{dockner1993international}, \citealp{nkuiya2021strategic}). Given the challenges of characterizing discontinuous strategies in heterogeneous differential games, such as our proposed model, Proposition \ref{prop1} and its associated proof are predicated on the strict assumption that the value function is differentiable on the entire state space. This regularity implies a continuous strategy function, allowing us to rigorously establish uniqueness of the optimal solution while circumventing the complexities arising from jump points in asymmetric assumptions.

\subsubsection{The Tipping Has Not Occurred\label{subsec:One-tipping-State-2}}
In State ``$pretip$'', the tipping event has not yet occurred, and the timing of the tipping event, denoted as $t_{1}$, is stochastic. To account for this uncertainty, we model the decision-making process of member and non-member regions using a stochastic dynamic game framework. Regions maximize their expected value functions, conditional on the stochastic nature of the tipping event. The optimization problems for coalition $\mathcal{S}$ and non-member regions ($j \notin \mathcal{S}$) are given by: 
\begin{equation}
\begin{aligned}
V^{\mathrm{pretip}}(T;\mathcal{S})
=
\max_{q_i(t),\, {i \in \mathcal{S}}}
\mathbb{E}
\sum_{i\in \mathcal{S}}
\Bigg\{
    & \int_{0}^{t_{1}} e^{-rt}
    \left[
        \left(1+\tau\right)\pi_i\left(q_i(t)\right)
        -D_i\left(T(t)\right)
    \right]dt
    \\
    & +\int_{t_{1}}^{\infty} e^{-rt}
    \left[
        \left(1+\tau\right)\pi_i\left(q_i(t)\right)
        -(1+\ell)D_i\left(T(t)\right)
    \right]dt
\Bigg\}.
\end{aligned}
\end{equation}

\begin{equation}
\begin{aligned}
W^{\mathrm{pretip}}_{j}(T;\mathcal{S})
=
\max_{q_{j}(t)}
\mathbb{E}
\Bigg\{
    & \int_{0}^{t_{1}} e^{-rt}
    \left[
        \pi_j\left(q_j(t)\right)
        -D_j\left(T(t)\right)
    \right]dt
    \\
    & +\int_{t_{1}}^{\infty} e^{-rt}
    \left[
        \pi_j\left(q_j(t)\right)
        -(1+\ell)D_j\left(T(t)\right)
    \right]dt
\Bigg\}.
\end{aligned}
\end{equation}

Lemma \ref{Lemma1} formalizes the equilibrium properties of value functions and optimal strategies for coalition members and non-member regions in State ``$pretip$''. 
\begin{lemma}\label{Lemma1}    
Before the occurrence of tipping, when $T\leq\underline{T}$, the pre-tipping value functions 
$V^{\mathrm{pretip}}(T;\mathcal{S})$ and 
$W^{\mathrm{pretip}}_j(T;\mathcal{S})$ are exactly quadratic in $T$, and the corresponding 
optimal emission strategies $q^{\mathrm{pretip}*}_i(T)$ and 
$q^{\mathrm{pretip}*}_j(T)$ are linear in $T$.

When $T>\underline{T}$, the hazard rate $H(T)$ makes the pre-tipping HJB system 
generally not closed under a finite-degree polynomial ansatz in the heterogeneous model. 
Hence, the pre-tipping value functions are approximated numerically.
\end{lemma}

\begin{proof}
The proof of Lemma \ref{Lemma1} is provided in Appendix \ref{appA3}. 
\end{proof}

Lemma \ref{Lemma1} establishes a sharp distinction across climate states. When the hazard term is inactive, the solution in the pre-tipping regime preserves the linear-quadratic structure and remains analytically tractable. Once $T>\underline{T}$, the hazard term generally destroys polynomial closure of the HJB system, requiring numerical approximation of the value functions. We present a numerical solution process to determine the equilibrium strategies for member and non-member regions prior to the climate tipping event. In State ``$tipped$'', the value functions $V^{\mathrm{tipped}}(T;\mathcal{S})$ and $W_j^{\mathrm{tipped}}(T;\mathcal{S})$ admit analytical solutions that incorporate the additional damages generated by the tipping event (see Appendix \ref{appA1}). In State ``$pretip$'', the solution depends on whether the tipping hazard is active. When $T\leq\underline{T}$, $H(T)=0$, and the value functions $V^{\mathrm{pretip}}(T;\mathcal{S})$ and $W_j^{\mathrm{pretip}}(T;\mathcal{S})$ retain the same linear-quadratic structure as in the post-tipping regime and can therefore be solved analytically. When $T>\underline{T}$, we therefore solve the pre-tipping value functions numerically using Chebyshev polynomial approximation, as described in Appendix \ref{appA4}\footnote{See \citet{judd1998numerical} and \citet{cai2019computational} for standard references on the numerical method.}, to ensure computational tractability while maintaining the rigor of the theoretical framework. 

The numerical solution yields not only the value functions $V$ and $W_j$, which summarize discounted future payoffs, but also the equilibrium utilities associated with the optimal emission strategies. Specifically, we compute $u_i(q_i^*,T(t))$ for each coalition member and $u_j(q_j^*,T(t))$ for each non-member by substituting the optimal emission decisions into their respective utility functions (e.g. Equations \ref{Eq.V_utility} and \ref{Eq.W_utility}). The corresponding equilibrium utilities characterize the contemporaneous payoffs generated by those optimal strategies. And the same definition of equilibrium utility applies in the following climate state with two tipping events.

\subsection{Decisions with Two Tippings}
\label{subsec:two tippings}
We extend our analysis to incorporate two sequential climate tipping events, where the occurrence of the second tipping event is conditional on the first. Let $\underline{T}_{1}$ and $\underline{T}_{2}$ denote the respective temperature thresholds at which the hazard rates for the first and second tipping events become non-zero. The hazard functions are specified as: 
$$H_{1}\left(T\right)=\chi_{1}\max\left(0,T-\underline{T}_{1}\right)$$ 
$$H_{2}\left(T\right)= 
\begin{cases}
0, & if\ the \ first\ tipping\ has\ not\ occurred \\
\chi_{2}\max\left(0,T-\underline{T}_{2}\right), & if\ the\ first\ tipping\ has\ occurred
\end{cases}$$
\noindent where $\chi_{1}$ and $\chi_{2}$ represent the hazard rate parameters for each tipping event. The associated climate damages are denoted by $\ell_1$ and $\ell_2$ for the first and second tipping events respectively. 

Let $t_{1}$ and $t_{2}$ represent the stochastic occurrence times of the first and second tipping events, with $t_{1}<t_{2}$. At any time $t$, the system may occupy one of three distinct states:

\textit{State ``$2tipped$": Both tipping events have occurred ($t \ge t_2$).} Denote by $V^{\mathrm{2tipped}}(T;\mathcal{S})$ and $W_j^{\mathrm{2tipped}}(T;\mathcal{S})$ the value functions of the coalition and each non‑member $j$, where the superscript ``$\mathrm{2tipped}$" indicates both tipping events have occurred.

\textit{State ``$pre2tip$": Only the first tipping event has occurred ($t_1 \le t < t_2$).} Denote by $V^{\mathrm{pre2tip}}(T;\mathcal{S})$ and $W_j^{\mathrm{pre2tip}}(T;\mathcal{S})$ the value functions of the coalition and each non‑member $j$, where the superscript ``$\mathrm{pre2tip}$" indicates the first tipping event has occurred but the second has not.

\textit{State ``$pre1tip$": No tipping event has occurred ($t < t_1$).} Denote by $V^{\mathrm{pre1tip}}(T;\mathcal{S})$ and $W_j^{\mathrm{pre1tip}}(T;\mathcal{S})$ the value functions of the coalition and each non‑member $j$, where the superscript ``$\mathrm{pre1tip}$" indicates no tipping event has yet occurred.

For State ``$2tipped$", the value functions remain quadratic in temperature $T$, allowing the analytical solutions to be derived using the exact methodology outlined in Appendix \ref{appA1}. For States 2.2 and 2.3, the value functions lose their quadratic structure, necessitating an adapted numerical framework based on recursive solution methods. Specifically, the numerical algorithm for ``$pre2tip$" is identical to that presented in Appendix \ref{appA4}. For State ``$pre1tip$", in which neither tipping event has yet occurred, we first obtain $V^{\mathrm{pre2tip}}$ and $W^{\mathrm{pre2tip}}$, which serve as continuation values, and then solve the HJB system for $V^{\mathrm{pre1tip}}$ and $W^{\mathrm{pre1tip}}$ using Chebyshev polynomial approximation (see Appendix \ref{appA5}).

The framework can be extended to more than two tipping events by introducing additional states that track the occurrence of successive tipping events. We therefore focus primarily on the single-tipping setting, which provides the basic structure for analyzing sequential climate thresholds and can be extended recursively to accommodate multiple tipping events.

\subsection{Payment Transfer and Coalition Stability}
\label{subsec:payment transfer}
Member regions within coalition $\mathcal{S}$ act collectively to maximize their expected discounted payoffs. Let $J$ index the tipping state. For a coalition comprising $m$ heterogeneous regions ($m \leq n$, where $m = |\mathcal{S}|$), we first derive the optimal emission strategy, $q_i^{J*}(T;\mathcal{S})$, for each member region $i \in \mathcal{S}$. Individual utilities are then determined through intra-coalition transfers, for which we consider two distinct allocation mechanisms:

\noindent \textbf{Mechanism I: Shapley Value Allocation}

\begin{equation}
\label{Eq.V_SV}
V_{i,SV}^{J}(T;\mathcal{S}) = \sum_{\mathcal{S}_{sub}\in \mathcal{S}(i)}\frac{\left(\left|\mathcal{S}_{sub}\right|-1\right)!\left(m-\left|\mathcal{S}_{sub}\right|\right)!}{m!}\left(V^J(T;\mathcal{S}_{sub})-V^J(T;{\mathcal{S}_{sub}\setminus\left\{ i\right\} })\right)
\end{equation}
\noindent where $V^J(T;\mathcal{S})$ denotes the value function of coalition $\mathcal{S}$ at temperature $T$, and $V_{i,SV}^{J}(T;\mathcal{S})$ represents the payoff allocated to member region $i \in \mathcal{S}$. Furthermore, $\mathcal{S}(i)$ defines the set of all sub-coalitions of $\mathcal{S}$ that include region $i$, with $\mathcal{S}_{sub} \in \mathcal{S}(i)$ representing a specific sub-coalition and $|\mathcal{S}_{sub}|$ denoting its cardinality. 

We use the Shapley values in Equation \ref{Eq.V_SV} to determine state-dependent sharing weights for the coalition's instantaneous equilibrium utility. Because these Shapley values are computed from the coalition value functions, the resulting weights reflect each member's marginal contribution to both current and future coalition payoffs. At each temperature and climate state, however, only the instantaneous coalition utility is allocated among members. Let $UV^J(T;\mathcal{S}) = \sum_{i\in\mathcal{S}}u_i(q_i^*,T;\mathcal{S})$ denote the equilibrium utility of coalition $\mathcal{S}$ at temperature $T$, and denote $UW_j^J(T;\mathcal{S}) = u_j(q_j^*,T;\mathcal{S})$ as the equilibrium utility function of non-member region $j$ at temperature $T$. Specifically, the utility allocated to member region $i$ is
\begin{equation}
\label{Eq.U_SV}
UV_{i,SV}^{J}(T;\mathcal{S})
=
UV^{J}(T;\mathcal{S})
\frac{V_{i,SV}^{J}(T;\mathcal{S})}
{V^{J}(T;\mathcal{S})}.
\end{equation}

\noindent Thus, Equation \ref{Eq.U_SV} uses each member's share of the coalition's intertemporal Shapley value as the weight for allocating instantaneous coalition utility. The Shapley calculation determines how the current payoff is divided across members, while the amount distributed at each state is given by the coalition's instantaneous equilibrium utility.

For the transferable-utility cooperative game, the Shapley value is the unique allocation rule satisfying efficiency, symmetry, additivity, and the dummy-player property (\citealp{hart1989potential}). Because the coalition value incorporates both instantaneous utilities (Equations \ref{Eq.V_utility} and \ref{Eq.W_utility}) and the intertemporal climate consequences of emissions, the resulting Shapley values reflect each member's marginal contribution to the coalition's intertemporal payoffs. We use the corresponding Shapley shares as state-dependent weights for allocating instantaneous coalition utility.

\noindent \textbf{Mechanism II: No Payment Transfer} 

This mechanism strictly precludes any inter-member financial transfers within the coalition. The value function for a member region $i \in \mathcal{S}$, denoted as $V_{i,NT}^{J}(T;\mathcal{S})$, is given by:
\begin{equation}
    V_{i,NT}^{J}(T;\mathcal{S}) = \mathbb{E}\int_{0}^{\infty}e^{-rt}u_i\left( q_i^*(t),T(t) \right)dt
\end{equation}

\noindent The utility for a member region $i \in \mathcal{S}$, denoted as $UV_{i,NT}^{J}(T;\mathcal{S})$, is given by:
\begin{equation}
UV_{i,NT}^J(T;\mathcal{S}) = u_i(q_i^*,T;\mathcal{S})
\end{equation}

In the numerical analysis, we compare coalition stability under the Shapley-value and no-transfer mechanisms. We use the Shapley allocation as our benchmark and apply the corresponding shares to instantaneous equilibrium utility. The no-transfer case provides a natural comparison. Quantitatively, the Shapley mechanism supports larger stable coalitions over the long run (see Section \ref{sec.allocation}).

\noindent \textbf{Coalition stability} 
\label{subsec:coalition stability}

Following the standard approach in IEA literature, we define a coalition's stability through two conditions:

\textit{Internal Stability}: 
$$
V_{i}^J\left(T;\mathcal{S}\right) \geq W_{i}^J\left(T;\mathcal{S}\setminus\left\{ i\right\}\right), \quad \forall i\in \mathcal{S}
$$

\textit{External Stability}: 
$$
V_{j}^J\left(T;\mathcal{S}\cup\left\{ j\right\}\right) \leq W_{j}^J\left(T;\mathcal{S}\right),\quad \forall j\notin \mathcal{S}\\
$$

These conditions ensure that: No member has incentive to leave (internal stability); No non-member has incentive to join (external stability). For the grand coalition $\bar{\mathcal S}$, external stability is irrelevant because there are no non-members.

\subsection{Coalition Dynamics}
\label{subsec:coalition dynamics}
Our framework characterizes the endogenous evolution of stable coalitions across climate states. At each temperature level, we compute all possible coalitions using temperature dependent value functions and state contingent payoffs, and then revaluate coalition stability. While some studies suggest that multiple agreements may enhance cooperation
(\citealp{asheim2006regional}), we retain a single-coalition framework, which is consistent with both major multilateral agreements such as the Paris Agreement and the benchmark climate-coalition literature (\citealp{nordhaus2015climate}). This restriction also avoids the equilibrium-selection issues that can arise in richer coalition-formation environments (\citealp{hyndman2007coalition}). 

To understand how coalition stability evolves with the climate state, we first examine how rising temperatures affect the relative emission incentives of coalition members and non-members. The following theorem identifies a temperature-driven free-riding channel in the absence of connection. Importantly, it does not require coalition members to emit less than non-members at every temperature. Rather, it shows that the relative emission advantage of remaining outside the coalition increases as temperature rises.

\begin{theorem}
\label{theorem.1}
Consider $n$ regions with identical payoff functions, divided into two groups: $m$ regions form a coalition and $n-m$ regions remain outside, where $n>m>1$. Regions are homogeneous within each group. In the post-tipping regime, suppose that the connection rate is zero, $\tau=0$. Then:
\begin{enumerate}
    \item For a fixed coalition of size $m$, the emission gap between a non-member and a coalition member increases with temperature.
    
    \item For a representative region, the emission gap between remaining outside the coalition and joining the coalition increases with temperature.
\end{enumerate}
\end{theorem}

\begin{proof}
The proof of Theorem \ref{theorem.1} is provided in Appendix \ref{appB1}.
\end{proof}

Theorem \ref{theorem.1} identifies a temperature-driven free-riding channel. As temperature rises, non-members retain progressively greater emission flexibility relative to coalition members. This widening emission gap can make participation more difficult to sustain, even when cooperation improves aggregate payoffs. Rising temperatures can therefore strengthen the tension between individual emission incentives and collective climate payoffs.

Beyond this gradual change in participation incentives, tipping introduces an additional source of nonlinearity. Both the threat of a tipping event and its occurrence can generate discrete adjustments in optimal emissions. The next theorem characterizes these changes under full cooperation (i.e., the grand coalition).

\begin{theorem}
\label{theorem.2}
Consider the grand coalition composed of heterogeneous regions and suppose that the coalition faces a tipping threshold $\underline{T}$. The threat and occurrence of a tipping event induce structural discontinuities in optimal emissions. Specifically:
\begin{itemize}
    \item[\textit{i.}] At the threshold ($T \to \underline{T}^{+}$), exposure to tipping induces a discrete precautionary reduction in optimal emissions relative to the no-tipping regime.
    
    \item[\textit{ii.}] The direction of the post-tipping emission jump depends on the balance between
    $(r+H(T))(V^{\mathrm{pretip}}-V^{\mathrm{tipped}})$
    and
    $\sum_{i\in\bar{\mathcal{S}}}\ell D(T(t))$.
    Emissions decrease if the former dominates, rebound if the latter dominates, and remain continuous if the two effects exactly offset each other.
\end{itemize}
\end{theorem}

\begin{proof}
The proof of Theorem \ref{theorem.2} is provided in Appendix \ref{appB2}.
\end{proof}

Theorem \ref{theorem.2} shows that tipping affects emissions through both anticipation and occurrence. Before tipping occurs, the prospect of future climate damages induces precautionary abatement. At the tipping event itself, however, the direction of the emission adjustment depends on the relative strength of the continuation-value loss and the increase in realized climate damages. These discontinuities can further alter the distribution of coalition utilities and, in turn, the incentives to remain in the coalition.

This interaction between climate tipping and participation incentives motivates the need for a mechanism that can strengthen the value of coalition membership. We focus on connection as an internal incentive mechanism. By raising the productivity of coalition members, connection can increase the attractiveness of participation, but it may also affect their optimal emission choices. We therefore examine how connection and tipping damages shape emissions under full cooperation.

\begin{theorem}
\label{theorem.3}
Consider heterogeneous regions under full cooperation, and suppose that aggregate marginal climate damages are non-negative. Holding the temperature state $T$ fixed:
\begin{enumerate}
    \item In both the no-tipping and post-tipping regimes, each region's optimal emissions increase as the connection rate increases and decrease as the connection rate decreases.
    \item In the post-tipping regime, each region's optimal emissions decrease as the tipping damage rate increases and increase as the tipping damage rate decreases.
\end{enumerate}
\end{theorem}

\begin{proof}
The proof of Theorem \ref{theorem.3} is provided in Appendix \ref{appB3}.
\end{proof}

Theorem \ref{theorem.3} identifies two distinct channels affecting the coalition's emission decision. Connection operates through the benefit side by increasing members' effective productivity, thereby raising optimal emissions under full cooperation. The tipping damage rate operates through the damage side: higher tipping damages increase temperature-related losses and reduce optimal emissions. These comparative statics are derived for a fixed coalition structure. When coalition membership is endogenous, however, the aggregate effect of connection also depends on its effect on coalition size. Our numerical results show that connection expands the stable coalition relative to the case without internal incentives or external sanctions. Given this increase in coalition membership, Theorem \ref{theorem.1} implies an offsetting emission effect, since a region emits less as a coalition member than when remaining outside. As a result, the reduction in emissions associated with broader participation more than offsets the direct emission-increasing effect of connection under a fixed coalition, leading to lower global emissions. This pattern is clearly visible in Figure \ref{Fig.4.3.1}.


In Section \ref{Sec.Results}, we translate the dynamic evolution of the coalition with respect to warming into dynamics over time. The results are derived by solving the model forward through the feedback between emissions, temperature, and coalition formation. At each time $t$, given the current temperature $T_t$, we compute the optimal emission strategies $q_i^*$ and $q_j^*$, evaluate the corresponding value functions $V$ and $W_j$, and determine the endogenous stable coalition $\mathcal{S}_t$. The resulting aggregate emissions are then used to update temperature according to the climate dynamics, yielding $T_{t+1}$. This procedure is repeated sequentially until the terminal time is reached. The corresponding equilibrium utilities are evaluated at each time step using the same optimal emission strategies and coalition configuration.

\section{Results}
\label{Sec.Results}

\subsection{Parametrization}

Model calibration draws on the regional economic and climate projections from the BAU (business as usual) solution of RICE-2010 (\citealp{nordhaus2010economic}), which comprises 12 aggregated regions. Following the standard distinction between advanced economies and emerging market and developing economies (\citealp{chari2014learning}), we group the model regions into advanced regions (US, EU \footnote{UK is included in EU for consistency with RICE-2010.}, Japan, and OHI (Other High Income regions)) and emerging \& developing regions (China, Russia, India, the Middle East (MidEast), Latin America (LatAm), Other Asia (OthAsia), Eurasia, and Africa). Regional GDP (converted to trillion PPP USD) and emissions (GtC) form the basis for estimating our benefit function parameters. We employ ordinary least squares to calibrate the continuous climate damage function coefficients, with connection effects benchmarked to 2025 values. The benefit function parameters $(\alpha_i,\beta_i)$ are calibrated through a three-step procedure (See Appendix \ref{appC1}). As a result, the calibrated parameters exhibit substantial regional heterogeneity, with advanced regions generally characterized by higher values of $\alpha_i$, while emerging \& developing regions display lower $\alpha_i$ and greater heterogeneity in $\beta_i$, implying distinct emission-benefit profiles across regions.

Key parameters include a 3\% discount rate ($r$), a 0.007 hazard rate ($\chi$), and a 0.0021 TCRE coefficient ($\lambda$) \citep{matthews2009proportionality, cai2025dynamics}. We consider a fixed membership coalition \citep{petrosjan2003time} instead of flexible membership \citep{rubio2007infinite}, since the model dynamically updates forward over time through the temperature-emission feedback loop. Besides, we choose a specific case setting as our benchmark, characterized by the tipping damage rate $\ell= 1$ and the connection rate $\tau = 0.6\%$. Multiple stable coalitions may emerge within a state during numerical calculations, we select the configuration that maximizes overall payoffs among regions. 

Moreover, the evaluation of two logically justified payment allocation mechanisms shows that the Shapley value allocation fosters cooperation more effectively (see Section \ref{sec.allocation}). Therefore our subsequent analysis focuses on the Shapley value allocation under one tipping event.

\subsection{Stable Coalition and Utility}
\label{Sec.Stable Coalition}

In the benchmark settings ($\ell=1,\underline{T}=2^\circ C,\bar \tau=0.6\%$), stable coalition progressively shrinks as temperature rises. This numerical pattern is consistent with the mechanism identified in Theorem \ref{theorem.1}: in the post-tipping regime, as temperature rises, the emission gap between coalition members and non-members widens, making the outside option relatively more attractive. The resulting increase in free-riding incentives weakens coalition stability (Figure \ref{Fig.4.2.1}).

\begin{figure}[h!]
\includegraphics[width=0.9\linewidth, keepaspectratio]{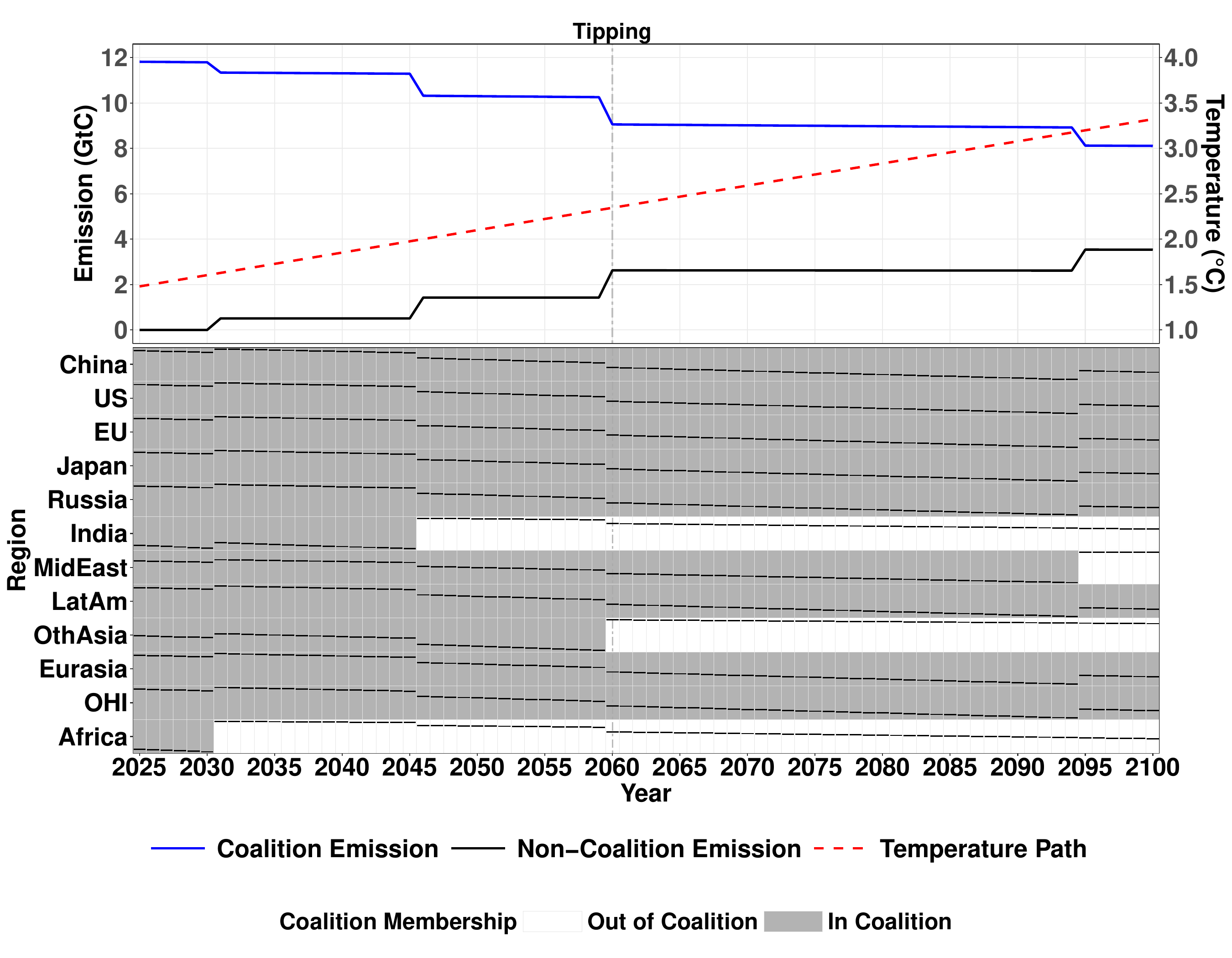}
\QTR{caption}{Stable coalitions and temperature rise from 2025 to 2100 under one tipping \label{Fig.4.2.1}}
\begin{figurenotes}
$\ell=1$,  $T_0=1.48^\circ C$, $\underline{T}=2^\circ C$, $\bar\tau=0.6\%$. The upper panel reports emissions from stable coalition members (blue line) and all non coalition regions (black line), together with the global temperature path represented by the red dashed line that follows the right vertical axis. The lower panel reports coalition membership by region. White cells indicate that a region is outside the coalition, while black cells indicate coalition membership. The line within each cell shows the region’s emissions normalized over 2025 to 2100, which is used to illustrate the trend in regional emissions over time. The gray vertical line denotes the simulated 2060 tipping event.
\end{figurenotes}
\end{figure}

The shrinkage of the stable coalition is concentrated among emerging \& developing regions, while advanced regions remain the coalition's core. Membership adjustments are most pronounced around the tipping event. India exits upon first exposure to the tipping threat, as temperature reaches the tipping threshold $\underline{T}=2^\circ\mathrm{C}$ around 2046, while OthAsia exits after the tipping event occurs. These membership changes are accompanied by substantial emission responses: regions increase emissions upon leaving and reduce them upon joining. Consistent with Theorem \ref{theorem.2}, both the threat and occurrence of tipping induce discrete changes in optimal emissions. The numerical results further show that both effects strengthen free-riding incentives among emerging \& developing regions.

This mechanism reflects the trade-off between the gains from cooperation and the benefit of the outside option. Coalition members gain from connection and from the lower climate damages generated by coordinated mitigation, but they also bear a larger share of the mitigation burden. Regions outside the coalition avoid this burden, retain greater scope for emissions, and continue to benefit from the abatement undertaken by coalition members. As climate conditions deteriorate, this trade-off becomes less favorable to membership for some emerging \& developing regions, generating the observed entry and exit dynamics.

\begin{figure}[h!]
\includegraphics[width=\linewidth, keepaspectratio]{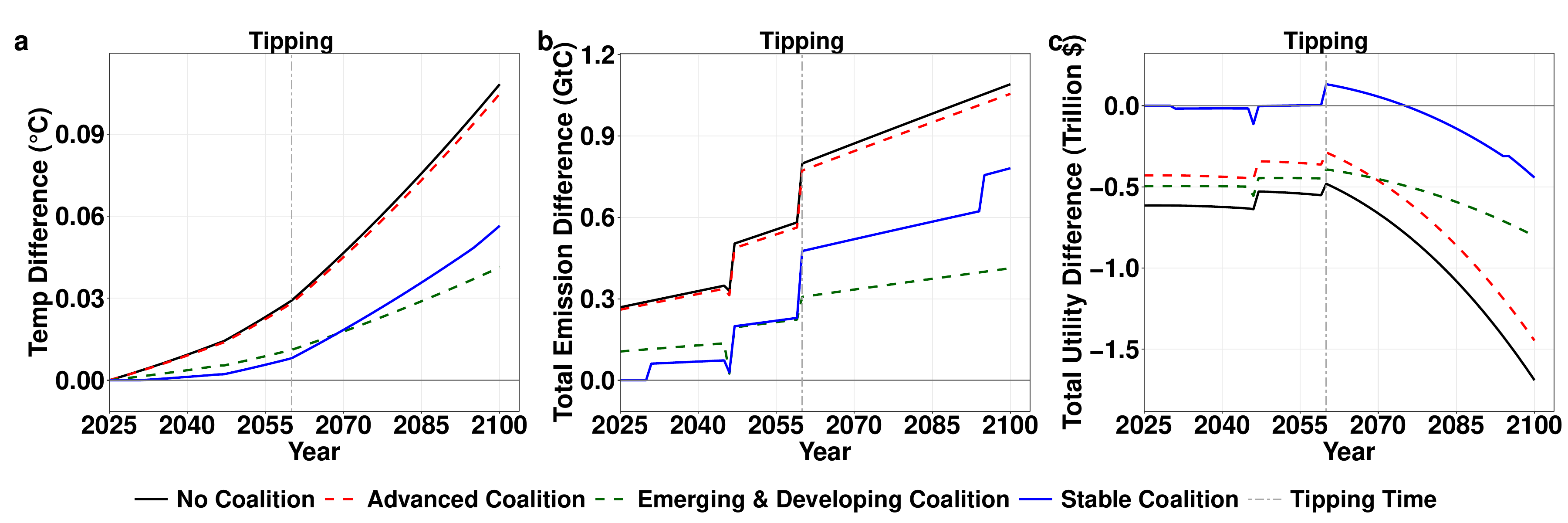}
\QTR{caption}{Temperature and total profit path under different coalitions \label{Fig.4.2.2}}
\begin{figurenotes}
$\ell=1$,  $T_0=1.48^\circ C$, $\underline{T}=2^\circ C$, $\bar\tau=0.6\%$. a \&b \& c. Differences relative to the grand coalition in temperature, global emissions, and instantaneous utilities, respectively. Emerging \& developing coalition consists of emerging \& developing regions (China, Russia, India, the Middle East, Latin America, Other Asia, Eurasia, and Africa); Advanced coalition means regions in the coalition are all advanced regions: US, EU, Japan, OHI. The gray vertical line denotes the 2060 tipping event.
\end{figurenotes}
\end{figure}

Figures \ref{Fig.4.2.2} shows that alternative coalition configurations differ substantially in their temperature, emissions, and utility implications by comparing with the grand coalition, which performs best across all three dimensions. A coalition comprising only advanced regions yields outcomes close to the noncooperation benchmark (i.e., ``No Coalition") in emissions and temperature. By contrast, a coalition formed exclusively by emerging \& developing regions lies much closer to the grand-coalition benchmark along these same dimensions. The stable coalition, which includes both advanced regions and some of emerging \& developing regions, performs slightly weaker than the emerging \& developing-only coalition in reducing emissions and limiting temperature increases, but generates substantially higher utilities. The temporary kink around 2046 reflects differences in the timing of the tipping threat: the grand coalition reaches the $2^\circ\mathrm{C}$ threshold one year later, delaying the associated adjustment in emissions and utilities. After 2075, the utility gap between the grand coalition and the stable coalition widens, illustrating the paradox of cooperation (\citealp{barrett1994self}): the gains from full cooperation increase even as sustaining such cooperation becomes more difficult.

\begin{figure}[h!]
\includegraphics[width=\linewidth, keepaspectratio]{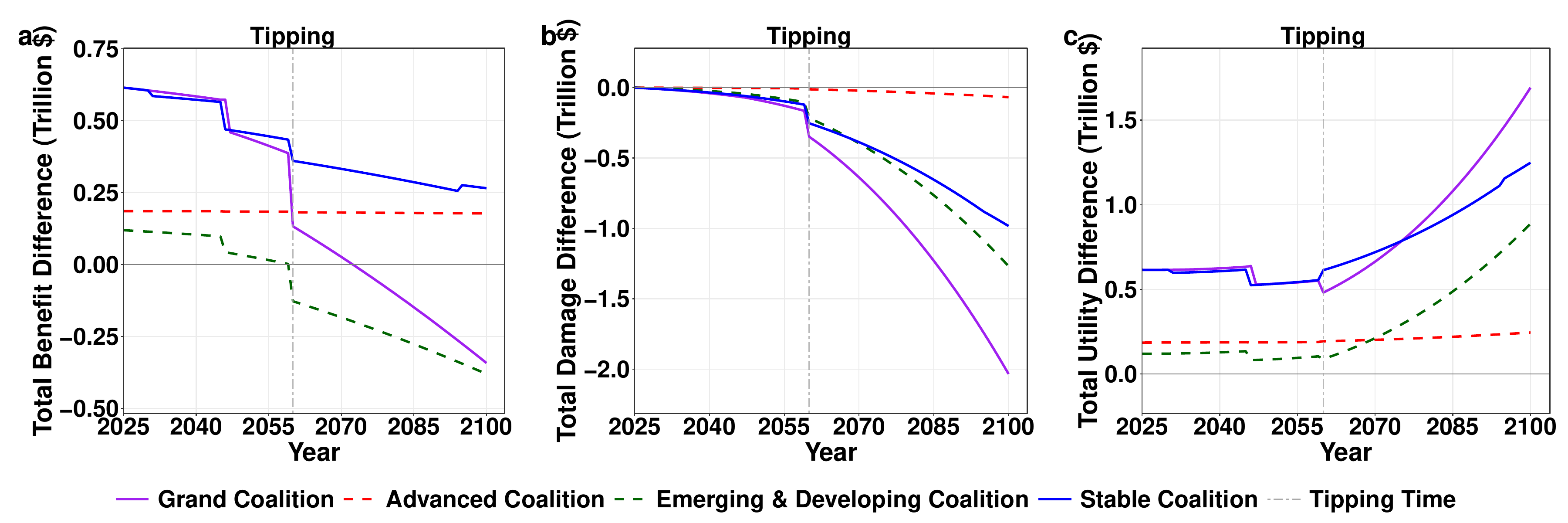}
\QTR{caption}{Total benefits and damage paths under different coalitions \label{Fig.4.2.3}}
\begin{figurenotes}
$\ell=1$,  $T_0=1.48^\circ C$, $\underline{T}=2^\circ C$, $\bar\tau=0.6\%$. a \&b \&c Differences relative to the noncooperation in total benefits, total damage, and total instantaneous utility, respectively. The gray vertical line denotes the 2060 tipping event.
\end{figurenotes}
\end{figure}

Comparison with the noncooperation highlights a clear trade-off across coalition structures (Figure \ref{Fig.4.2.3}). A coalition composed only of advanced regions gains primarily through connection benefits. The grand coalition and the coalition formed by emerging \& developing regions retain lower climate damages throughout the period, but their benefit advantage disappears after tipping. The stable coalition, by contrast, delivers both higher aggregate benefits and lower climate damages than noncooperation. Compare to the grand coalition, it sacrifices part of the damage reduction but avoids the sharp post-tipping decline in benefits.

Overall, the stable coalition reflects a trade-off between climate outcomes and participation incentives across advanced regions and emerging \& developing regions. Relative to the grand coalition, it delivers smaller reductions in emissions and climate damages, but preserves a larger share of the benefits associated with coalition participation. Relative to coalitions formed predominantly by either advanced regions or emerging \& developing regions, it achieves a more balanced combination of climate outcomes and aggregate utility. The stable coalition thus balances the gains from broader mitigation against the incentives required to sustain participation across heterogeneous regions.

\subsection{Connection Versus Sanction}
\label{Sec.Tech VS Sanc}
We extend the sanction mechanism proposed by \citet{nordhaus2015climate} to our dynamic setting with tipping. Under this mechanism, non-member regions face a sanction rate $\tau_2$, in contrast to the connection benefits offered to coalition members. For instance, in State ``$tipped$'', the collective payoff of coalition $\mathcal{S}$ is:

\begin{equation}
\max_{q_e(t), {e \in \mathcal{S}}}\ \sum_{e\in \mathcal{S}}\int_{0}^{\infty}e^{-rt}\left(\pi_e\left(q_e(t)\right)-(1+\ell)D_e\left(T(t)\right)\right)dt
\end{equation} 

\noindent while the payoff for each non-member region $f\notin \mathcal{S}$ is
\begin{equation}
\max_{q_f(t)}\int_{0}^{\infty}e^{-rt}\left(
(1-\tau_2)\pi_f(q_f(t))-(1+\ell)D_f(T(t))
\right)dt
\end{equation}

Figure \ref{Fig.4.3.1} reports a clear trade-off between the two instruments with the same connection and sanction rates ($\bar{\tau}=\bar{\tau}_2=0.6\%$). Combining connection and sanctions yields the largest stable coalition, the highest aggregate payoff, and the lowest global emissions. When implemented separately, the two mechanisms yield similar coalition sizes and emissions reductions: sanctions support a slightly larger stable coalition, whereas connection delivers a substantially higher aggregate payoff. Theorem~\ref{theorem.1} shows that each region emits less as a coalition member than as a non-member, suggesting that differences in coalition size help explain the observed differences in global emissions.

The emissions gap between connection and sanctions varies across tipping states and is particularly pronounced around tipping. The threat of tipping induces larger emissions reductions under sanctions than under connection. After tipping, the gap widens further as emissions continue to decline under sanctions but rise under connection. This pattern reflects both the participation effect in Theorem~\ref{theorem.1} and non-members' stronger response to external sanctions as climate conditions worsen. However, the benefits of sanctions depend on their costs. When coalition members bear part of the losses imposed on non-members, as emphasized in the sanctions literature \citep{itskhoki2026sanctions,grossman2024tariffs}, sanctions' benefits for coalition formation and mitigation diminish, and emissions rise. Thus, at a given policy intensity, sanctions can achieve greater mitigation than connection, especially around tipping, but at a lower aggregate payoff and conditional on enforcement costs.

\begin{figure}[h!]
\includegraphics[width=\linewidth, keepaspectratio]{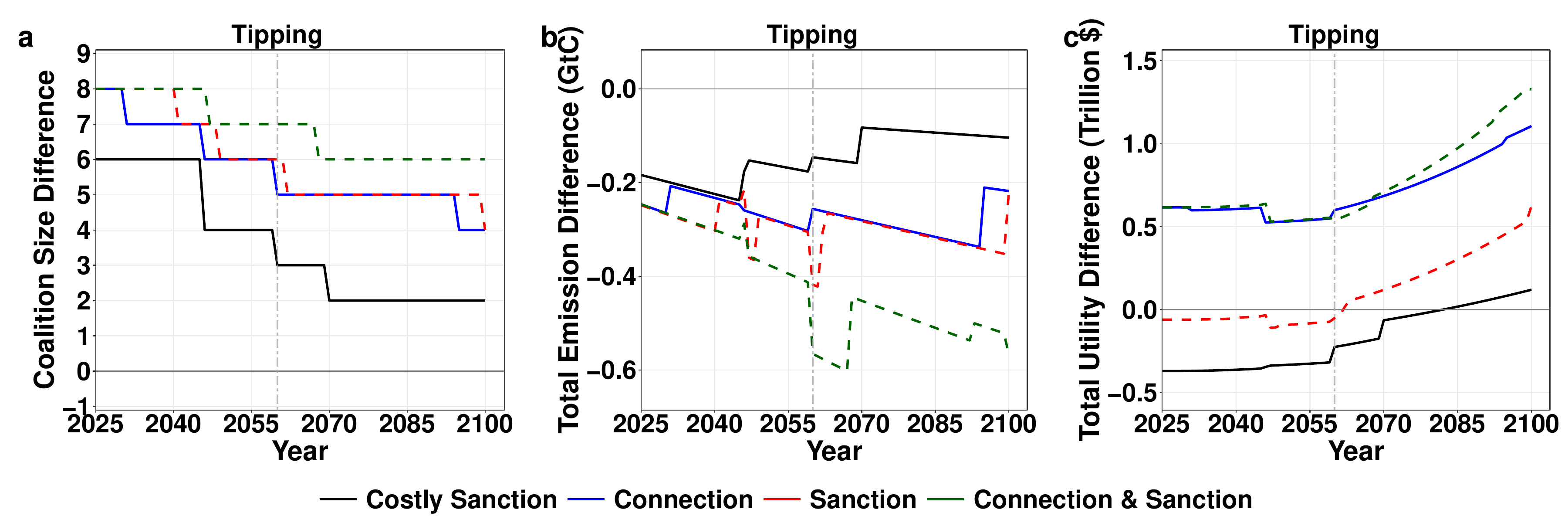}
\QTR{caption}{Climate coalition dynamics under Shapley value allocation with one tipping event and different connection /sanction rates \label{Fig.4.3.1}}
\begin{figurenotes}
$\ell=1$,  $T_0=1.48^\circ C$, $\underline{T}=2^\circ C$, $\bar\tau=\bar\tau_2=0.6\%$.
a. Differences in stable coalition size between each policy case and the benchmark without connection or sanctions.
b. Differences in global emissions between each policy case and the benchmark without connection or sanctions.
c. Differences in total instantaneous utility between each policy case and the benchmark without connection or sanctions. Under the ``costly sanction" specification, coalition members bear 50\% of the sanctions imposed on non-members; the ``connection \& sanction" specification combines internal connection with external trade sanctions. The gray vertical line denotes the 2060 tipping event.
\end{figurenotes}
\end{figure}

Connection and sanctions both reduce emissions and can sustain the grand coalition relative to no policy (i.e., $\tau=\tau_2=0$), though they respond differently to tipping. How strong must each instrument be as climate conditions worsen? Based on the collective rationality condition, we compare the minimum policy rates required to sustain the grand coalition. Define the stability gap as
\begin{equation}
\Psi(\bar{\tau},T)
= V(\bar{\tau},T;\bar{\mathcal S})-\sum_{i=1}^{n}W_i\left(T;\bar{\mathcal S}\setminus\{i\}\right).
\end{equation}

The minimum connection rate satisfies $\Psi(\hat{\tau}(T),T)=0$, and
the minimum sanction rate is obtained analogously. We compute these minimum rates at each temperature along the equilibrium path to measure the policy intensity required to maintain collective rationality and sustain the grand coalition as temperature rises.

\begin{figure}[h!]
\includegraphics[width=0.5\linewidth, keepaspectratio]{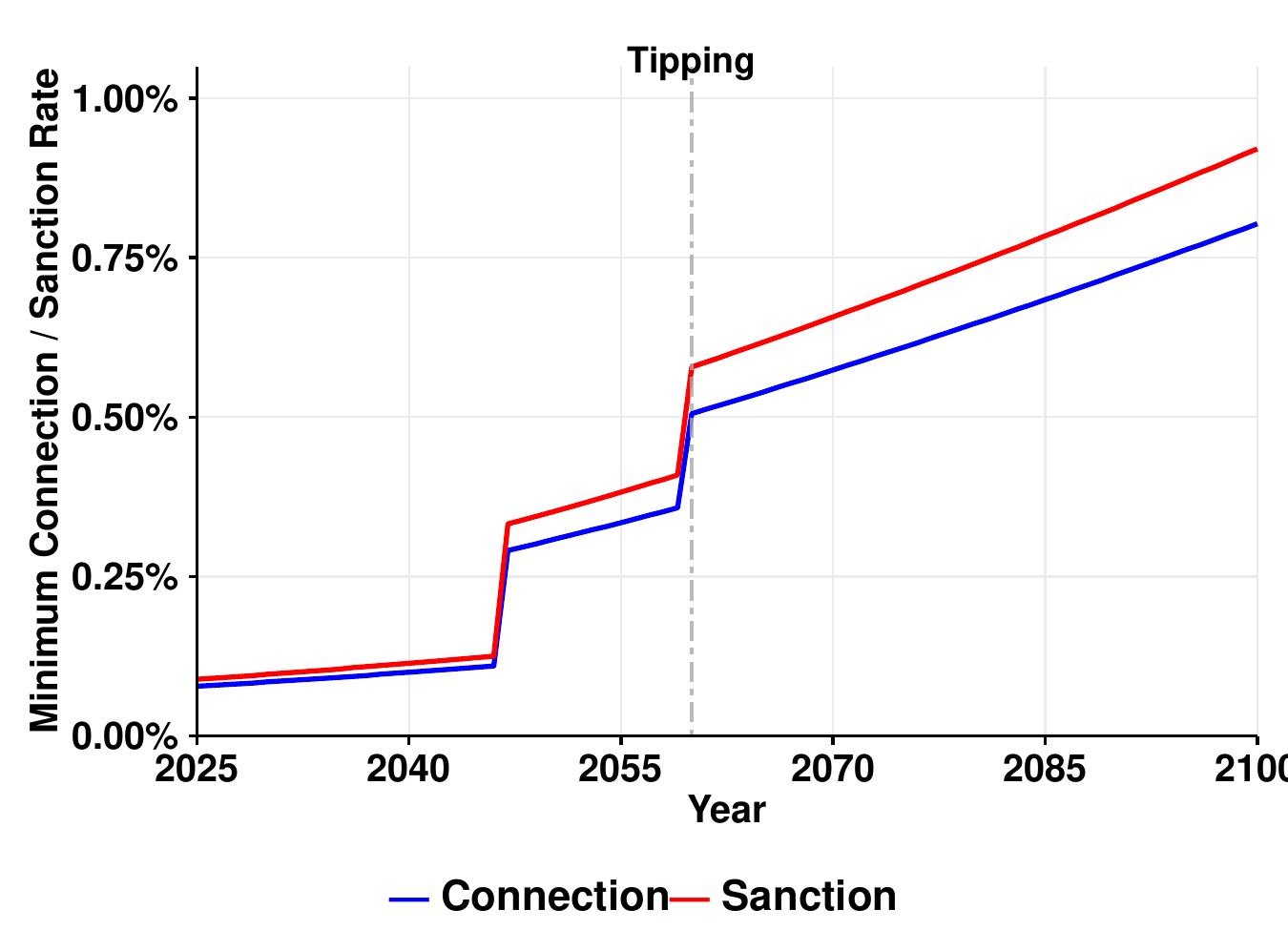}
\QTR{caption}{Minimum connection rate ($\hat{\tau}$) and sanction rates ($\hat{\tau_2}$) for grand coalition stability \label{Fig.4.3.2}}
\begin{figurenotes}
$\ell=1$,  $T_0=1.48^\circ C$, $\underline{T}=2^\circ C$. The gray vertical line denotes the 2060 tipping event. 
\end{figurenotes}
\end{figure}

Figure \ref{Fig.4.3.2} shows that the minimum policy rates rise with temperature and exhibit discrete jumps at both the threat and occurrence of tipping. The sanction rate remains above the corresponding connection rate throughout, and the gap widens with temperature. Thus, a tipping event raises the policy effort required to sustain full cooperation, as external enforcement becomes progressively more demanding relative to internal incentives. The discrete changes around tipping mirror the emission responses characterized in Theorem \ref{theorem.2}. From this perspective, the minimum connection and sanction rates measure more than policy intensity: they summarize the institutional effort required to offset rising free-riding incentives and preserve the grand coalition as temperature rises and tipping becomes more imminent.

\subsection{Mechanisms for Utility Allocation}
\label{sec.allocation}

We compare two allocation schemes: \textit{Shapley value allocation}, which redistributes coalition surplus according to each member's marginal contribution, and \textit{No payment transfer}, under which utilities are not reallocated. The former represents an ideal cooperative benchmark; the latter mirrors the institutional reality of most international climate agreements, which lack enforceable side payments. The comparison is conducted under a single tipping event and focuses on how the allocation rule affects the evolution of stable coalition membership. Figure \ref{Fig.4.4} shows that the difference between the two grows with temperature. At low temperatures, no payment transfer performs reasonably well, especially in the short run. As the temperature approaches the tipping threshold, however, free‑riding incentives intensify and undermine coalition stability. Shapley value allocation then supports a larger and more persistent coalition by improving the distribution of coalition utilities across members. Consequently, the Shapley rule becomes increasingly effective in sustaining cooperation as temperature rises and tipping threat becomes more salient.

\begin{figure}[h!] 
\includegraphics[width=\linewidth, keepaspectratio]{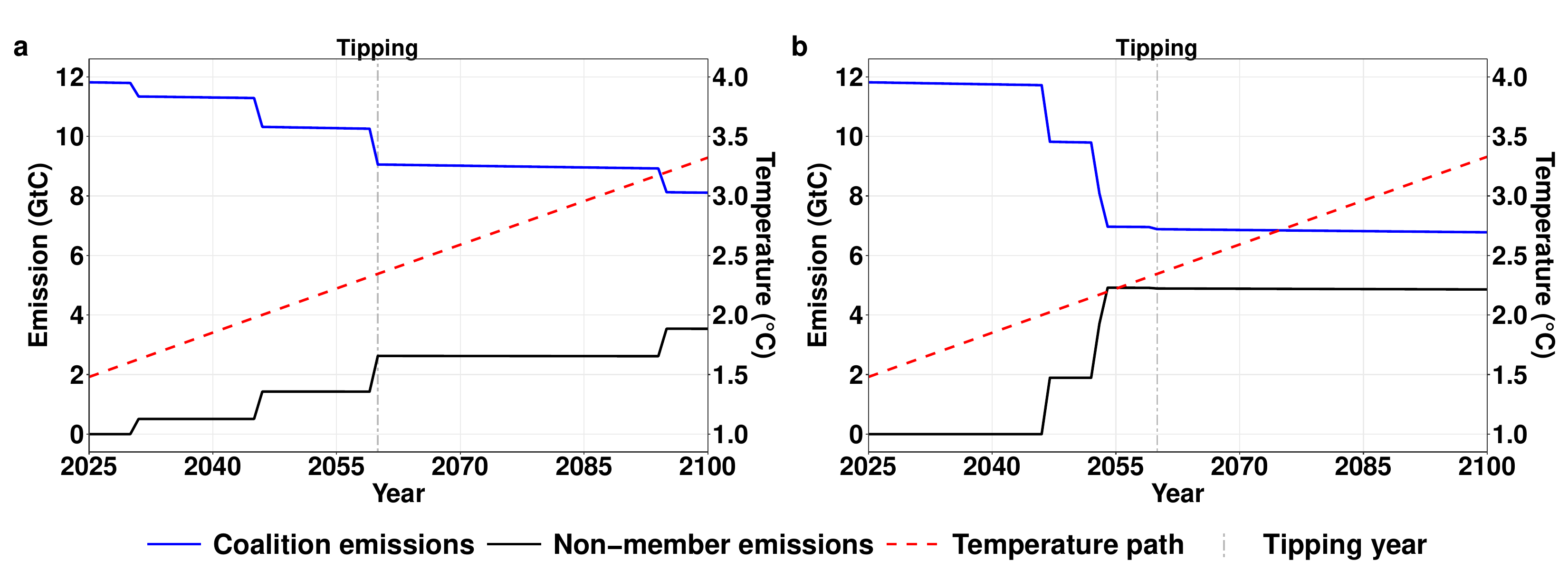}\
\QTR{caption}{Stable coalition emissions and temperature rise from 2025 to 2100 under different allocation types\label{Fig.4.4}}
\begin{figurenotes}
$\ell=1$,  $T_0=1.48^\circ C$, $\underline{T}_2=2^\circ C$, $\bar\tau=0.6\%$. The gray vertical line denotes the 2060 tipping event. a. Results based on Shapley value allocation; b. Results under no payment transfer.
\end{figurenotes}
\end{figure}

\subsection{Model Sensitivity Analysis}
\label{sec.sensitivity}

We assess robustness along four dimensions: first, we compare the effects of different connection rates $\tau$ and tipping damage rates $\ell$.
Second, we examine how the results vary across alternative benefit functions. Third, we assess the implications of alternative damage functions. Fourth, we compare the outcomes under two tipping events with those under a single tipping event. These checks confirm that our core results are robust to alternative specifications.

\subsubsection{Sensitivity to Connection and Tipping Damage}
\label{Sec.SA of connnection and Tipping Damage}

We vary one parameter at a time relative to the benchmark. Specifically, the connection rate increases from $0.6\%$ to $1\%$ and the tipping damage rate from $1$ to $1.5$. holding all other parameters unchanged.

\begin{figure}[h!]
\includegraphics[width=0.8\linewidth, keepaspectratio]{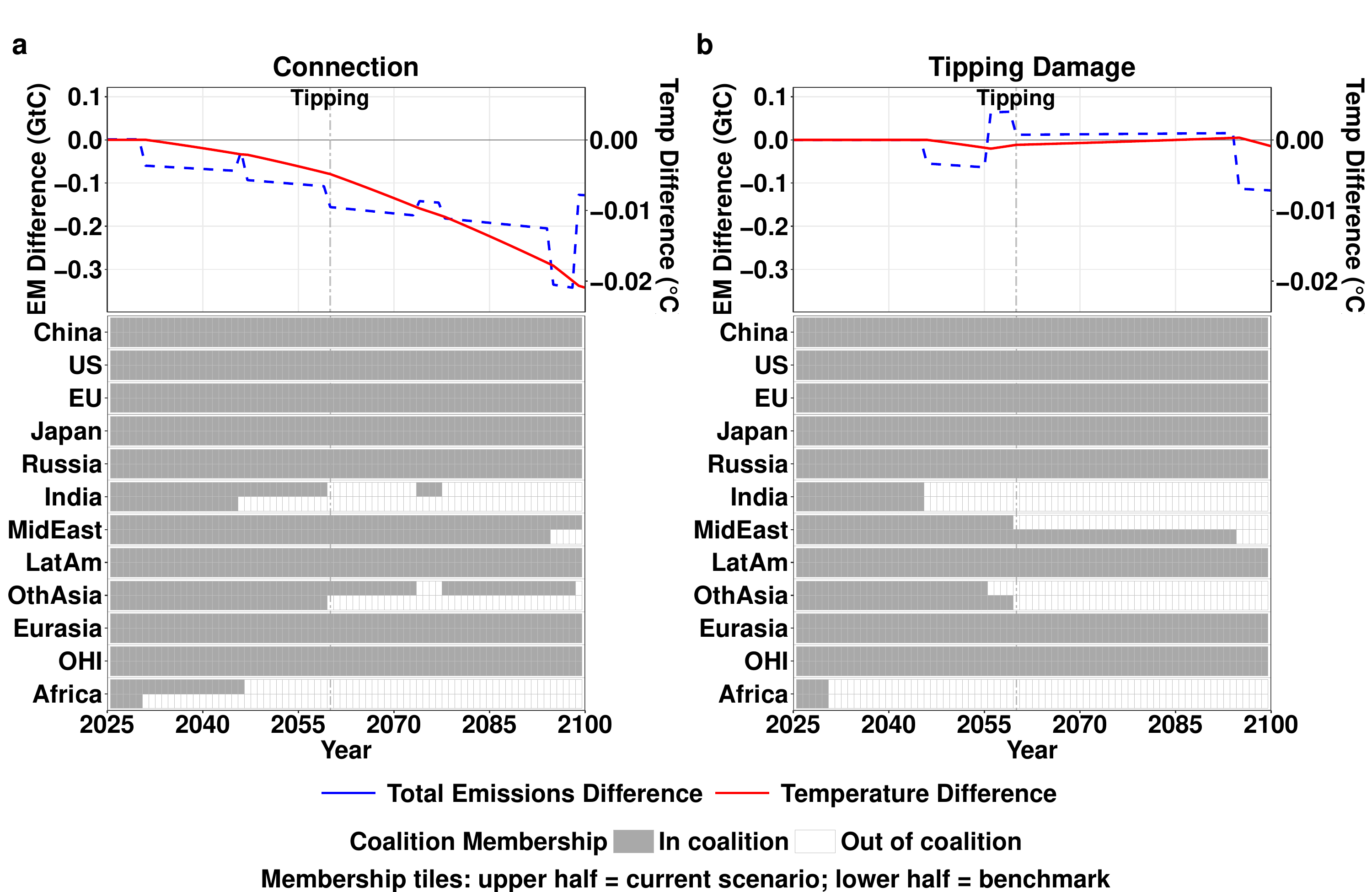}
\QTR{caption}{Differences in temperature and emission between various cases and benchmark\label{Fig.4.5.1}}
\begin{figurenotes}
a. $\tau=1\%$, other parameters remain as in the benchmark case; b. $\ell=1.5$, other parameters remain as in the benchmark case. In each panel, the upper part shows the differences in temperature and emissions between the current case and the benchmark, while the lower part shows regional coalition membership. In the membership plot, the upper half of each cell represents coalition membership under the current case, and the lower half represents coalition membership under the benchmark. Gray denotes membership, and white denotes non membership.
\end{figurenotes}
\end{figure}

Figure \ref{Fig.4.5.1}a shows that a higher connection rate enlarges coalition membership by attracting several emerging \& developing regions, including India, MidEast, Other Asia, and Africa, which reduces global emissions and slows temperature growth when the stable coalition is not the grand coalition in either case. However, in the initial years with the grand coalition in both cases (i.e., 2025-2030), a higher connection rate slightly increases global emissions. This is consistent with Theorem \ref{theorem.3} under full cooperation and State ``$notip$", while the difference after 2030 is mainly driven by the coalition composition change between two cases. Figure \ref{Fig.4.5.1}b shows that a higher tipping damage rate operates differently. Around 2046, when coalition composition is unchanged, the stronger tipping threat produces a larger reduction in emissions. As temperature continues to rise, however, some regions leave earlier: OthAsia exits in 2056 and MidEast in 2060, partly offsetting the initial mitigation response through higher emissions outside the coalition. This pattern follows Theorem \ref{theorem.2}.

\subsubsection{Sensitivity to Benefit Function}
\label{Sec.SA of benefit function}
To examine whether the shrinkage of the stable coalition depends on the baseline benefit specification, we replace the quadratic emission benefit function with quartic form: 

$$\pi(t) = \alpha q(t) - \frac{1}{4}\beta q(t)^4$$

We then recalibrate the model under the quartic benefit specification and report the resulting parameter values in Appendix \ref{appC1}. Using these parameters, we recompute coalition membership, emissions, and temperature under the Shapley value allocation. 
\begin{figure}[h!]
\includegraphics[width=0.9\linewidth, keepaspectratio]{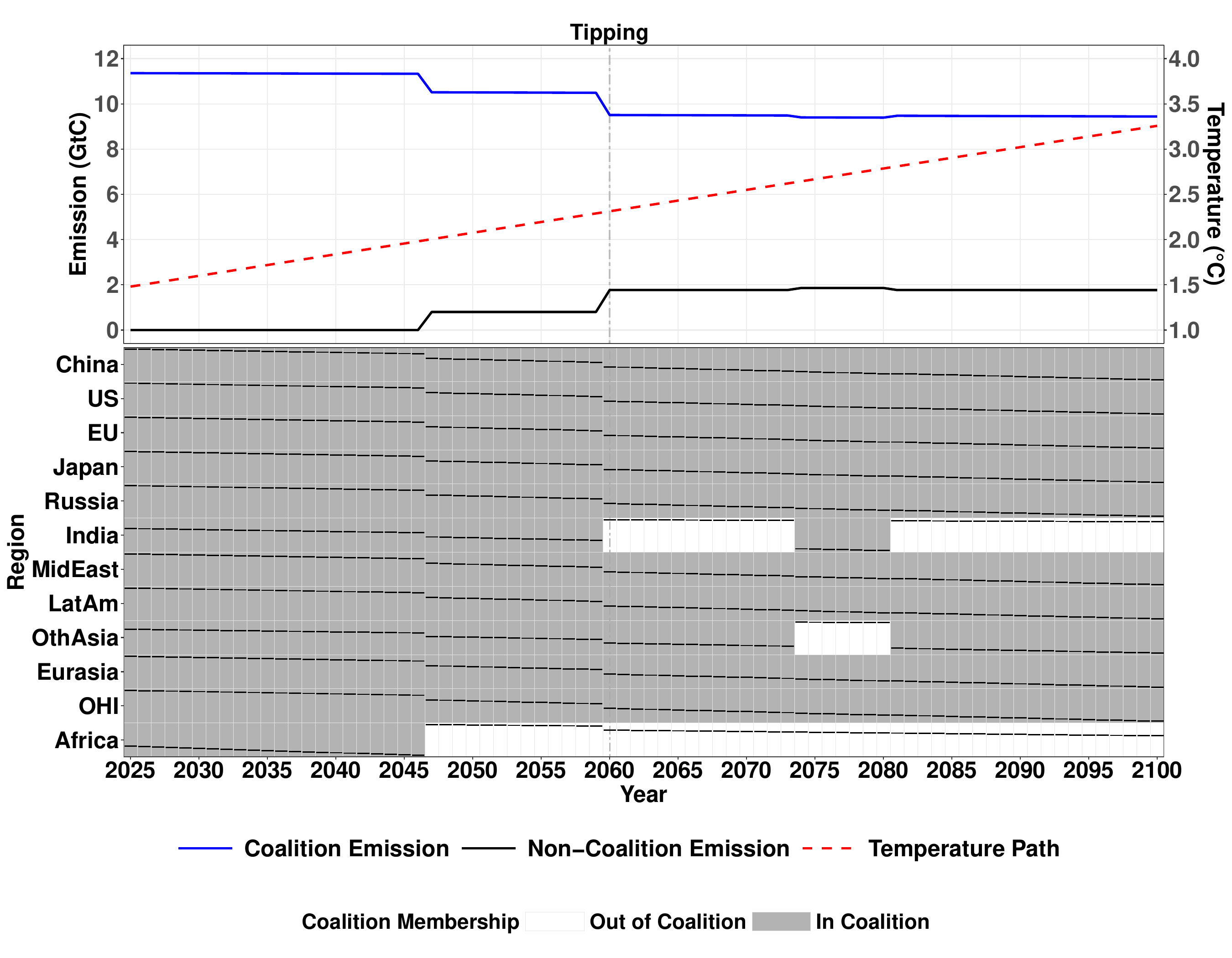}
\QTR{caption}{Sensitivity analysis of benefit function: Stable coalitions and temperature rise from 2025 to 2100 under one tipping\label{Fig.4.5.2}}
\begin{figurenotes}
$\ell=1$,  $T_0=1.48^\circ C$, $\underline{T}=2^\circ C$, $\bar\tau=0.6\%$. The gray vertical line denotes the 2060 tipping event.
\end{figurenotes}
\end{figure}

Figure \ref{Fig.4.5.2} shows that the main pattern remains unchanged. Advanced regions stay in the coalition throughout, while membership adjustments remain concentrated among emerging \& developing regions such as India, OthAsia, and Africa. The stable coalition continues to shrink as temperature rises and around both the threat and occurrence of tipping, consistent with the mechanisms in Theorems \ref{theorem.1} and \ref{theorem.2}. Relative to the benchmark, however, the quartic benefit specification supports a larger stable coalition, higher coalition emissions, and a faster increase in temperature.

\subsubsection{Sensitivity to Damage Function}
\label{Sec.SA of damage function}

To test whether the decline in coalition size is sensitive to the damage function, we also consider an alternative damage specification:

$$D_i(T)=\eta_i T+\frac{1}{4}\gamma_i T^4,$$

Holding the emission benefit function and coalition formation procedure fixed, we recompute coalition membership, emissions, and temperature under the Shapley value allocation. 

\begin{figure}[h!]
\includegraphics[width=0.9\linewidth, keepaspectratio]{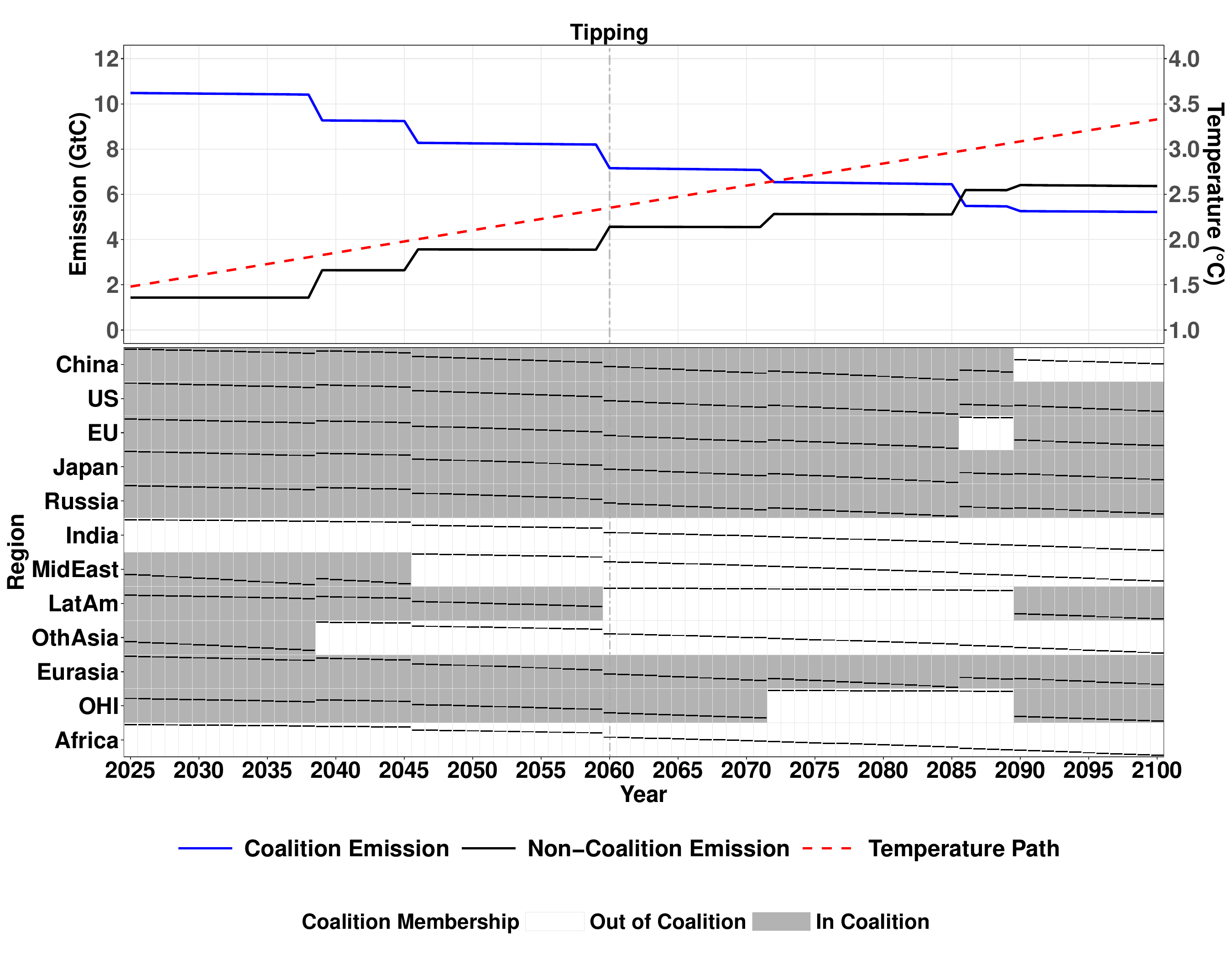}
\QTR{caption}{Sensitivity analysis of damage function: Stable coalitions and temperature rise from 2025 to 2100 under one tipping\label{Fig.4.5.3}}
\begin{figurenotes}
$\ell=1$,  $T_0=1.48^\circ C$, $\underline{T}=2^\circ C$, $\bar\tau=0.6\%$. The gray vertical line denotes the 2060 tipping event.
\end{figurenotes}
\end{figure}

Figure \ref{Fig.4.5.2} shows that the quartic damage specification alters the composition of the stable coalition but preserves its dynamics over time. Compared with the benchmark, some advanced regions, including the EU and OHI, exit as temperature rises but later rejoin the coalition, while emerging \& developing regions remain particularly responsive around the threat and occurrence of tipping. MidEast exits around 2046, and LatAm follows after the tipping event in 2060. Coalition emissions nevertheless continue to decline with rising temperature and around tipping, consistent with Theorems \ref{theorem.1} and \ref{theorem.2}. Overall, the stronger temperature sensitivity of damages amplifies free-riding incentives among both advanced and emerging \& developing regions. 

\subsubsection{Sensitivity to Tipping Events}
\label{Sec.SA of tipping events}

We extend the benchmark to a setting with two tipping events to assess whether the mechanisms identified under a single tipping event persist when tipping occurs repeatedly over time.

\begin{figure}[h!]
\includegraphics[width=\linewidth, keepaspectratio]{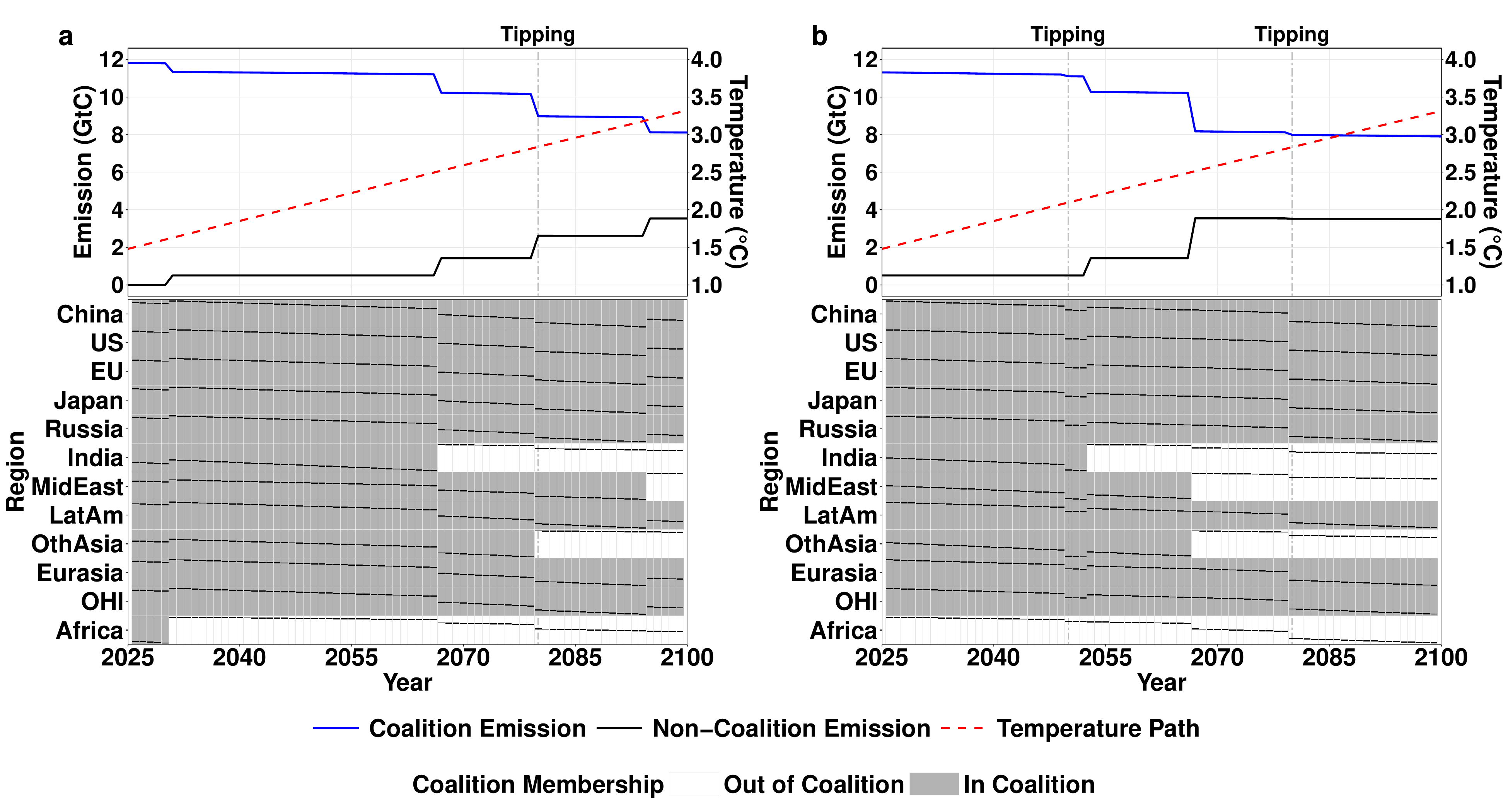}
\QTR{caption}{Sensitivity analysis of tipping events: Stable coalitions and temperature rise from 2025 to 2100 \label{Fig.4.5.4}}
\begin{figurenotes}
In one-tipping settings, $\ell=1$,  $T_0=1.48^\circ C$, $\underline{T}=2.5^\circ C$, $\bar\tau=0.6\%$. The gray vertical lines denote the 2080 tipping event.
In two-tippings settings, $\ell_1=0.5$, $\ell_2=1$,  $T_0=\underline{T_1}=1.48^\circ C$, $\underline{T_2}=2.5^\circ C$, $\bar\tau=0.6\%$. The gray vertical lines denote the 2050 and 2080 tipping events.
\end{figurenotes}
\end{figure}

Figure \ref{Fig.4.5.4} shows that a second tipping event does not overturn the core patterns. Coalition emissions decline as temperature rises and fall even more sharply around both the threat and occurrence of tippings, consistent with Theorems \ref{theorem.1} and \ref{theorem.2}. Advanced regions remain in the coalition throughout, while membership adjustments are concentrated among emerging \& developing regions, including India, MidEast, and OthAsia. With two tipping events, these regions continue to account for most of the observed entry and exit, indicating that repeated tipping strengthens free-riding incentives almost exclusively among this group.

\section{Conclusions}
\label{Sec.Conclusions}

This paper places climate coalition stability in a setting in which the climate state is both the outcome of strategic emissions and a determinant of the incentives to cooperate. The model characterizes stable coalitions conditional on each temperature state, while the numerical analysis traces how the corresponding coalition composition varies along the warming path. The principal finding is that increasing climate damage does not translate monotonically into stronger cooperation. Stable coalitions generally shrink as temperature rises: advanced regions remain their persistent core, whereas emerging \& developing regions constitute the principal margin of adjustment. This composition is consequential because regions face different emission incentives inside and outside the coalition.

Climate tipping sharpens this underlying tension. The threat of tipping strengthens incentives for precautionary mitigation but also encourages free-riding on coalition members' mitigation efforts. In our calibrated setting, the latter effect dominates, causing the tipping threat to raise aggregate emissions. The occurrence of tipping can increase or decrease emissions, yet its effect on cooperation is less ambiguous in the numerical results: both the threat and the occurrence of tipping reduce coalition size and intensify free-riding. Greater demand for collective action may therefore coincide with a narrowing of the set of coalitions capable of sustaining it.

Connection and sanctions address this tension through opposing margins: the former raises the payoff of membership, while the latter raises the cost of non-participation. Both instruments can support coalitions of similar size, yet this quantitative parity masks a substantial divergence in economic performance. Connection yields higher aggregate payoffs, whereas sanctions achieve larger emissions reductions. Thus, neither instrument strictly dominates the other. More generally, coalition size is not a sufficient statistic for the success of climate institutions: a coalition sustained by expanding the gains from membership is economically distinct from one sustained by compressing the outside option, even when the resulting participation patterns are similar.

These findings overturn, rather than qualify, the presumption that an increasingly adverse climate will, by itself, bring national incentives into closer alignment with the collective interest. For the Paris Agreement, the central institutional problem is consequently not limited to obtaining broad participation or progressively strengthening stated commitments. It is also to maintain the gains from cooperation as temperature and tipping alter the strategic position of members and nonmembers. Our analysis treats connection and sanctions in reduced form and compares stable coalitions across climate states rather than deriving an endogenous path of accession and withdrawal. Embedding these mechanisms in a setting with bargaining, renegotiation, and instrument credibility is a natural direction for further research. The Paris Agreement has established a common destination; whether cooperation can be sustained along the path to it remains fundamentally an institutional design problem—one that must adapt as the climate itself evolves.

\bibliographystyle{Templates/aea}
\bibliography{IEA_AER_Main}

\newpage

\clearpage
\setcounter{page}{1}

\appendix

\renewcommand{\thesubsection}{\arabic{subsection}}

\global\long\def\thefigure{A\arabic{figure}}%
\setcounter{figure}{0} 
\global\long\def\thetable{A\arabic{table}}%
\setcounter{table}{0} 
\global\long\def\theequation{A\arabic{equation}}%
\setcounter{equation}{0} 

\section{Appendix A: Analytical Characterization and Numerical Solution}\label{appA}
\subsection{Proof of Proposition \ref{prop1}}\label{appA1}

Since the benefit function $\pi(q(t))$ is strictly concave in emissions $q(t)$ and the damage function $D(T(t))$ is strictly convex in temperature $T(t)$, the instantaneous utility
\[
(1+\tau)\pi(q(t))-(1+\ell)D(T(t))
\]
is jointly concave in $(q(t),T(t))$, provided that $1+\tau>0$ and $1+\ell>0$. Accordingly, the discounted objective functional
\[
\int_{0}^{\infty}e^{-rt}
\left[
(1+\tau)\pi(q(t))-(1+\ell)D(T(t))
\right]dt
\]
is concave over admissible state-control paths.

Furthermore, the state dynamics,
\[
\frac{dT(t)}{dt}
=
\lambda\sum_i q_i(t),
\]
are linear in the control variables, so the feasible set of state-control paths is convex. Standard results from optimal control theory therefore imply that the associated value function is concave in the state variable $T$.

We derive the optimal strategies for regions within coalition $\mathcal{S}$
by maximizing their collective payoffs:
\begin{equation}
\max_{q_i(t), {i\in\mathcal S}}
\sum_{i\in\mathcal S}\int_{0}^{\infty}e^{-rt}
\left[(1+\tau)\pi_i(q_i(t))-(1+\ell)D_i(T(t))\right]dt 
\end{equation}
where temperature evolution follows: 
\begin{equation}
\frac{dT\left(t\right)}{dt}=\lambda\left(\sum_{i\in \mathcal{S}}q_{i}(t)+\sum_{j\notin \mathcal{S}}q_{j}(t)\right)
\end{equation}

By deriving the coalition's associated Hamilton-Jacobi-Bellman (HJB) equation, the value function satisfies: 
\begin{equation}
\begin{aligned}
rV^{\mathrm{tipped}}\left(T\right)
=
\max_{q_i(t),\, {i\in \mathcal{S}}}
\Bigg[
    & \sum_{i\in \mathcal{S}}
    \left(1+\tau\right)
    \left(
        \alpha_{i}q_i(t)
        -\frac{1}{2}\beta_{i}{q_i(t)}^{2}
    \right)
    \\
    & -\left(1+\ell\right)
    \sum_{i\in \mathcal{S}}
    \left(
        \frac{1}{2}\rho_{i}T(t)^{2}
        +\eta_{i}T(t)
    \right)
    \\
    & +\lambda
    \frac{dV^{\mathrm{tipped}}\left(T\right)}{dT}
    \left(
        \sum_{i\in \mathcal{S}}q_i(t)
        +\sum_{j\notin \mathcal{S}}q_j(t)
    \right)
\Bigg].
\end{aligned}
\end{equation}

\noindent where $dV^\mathrm{tipped} \left(T\right)/dT$ is the first-order derivative of the value function $V^{\mathrm{tipped}}(T)$ with respect to the state variable $T$. 

The first-order conditions yield optimal emissions for coalition members: 
\begin{equation}
q_{i}^{\mathrm{tipped}*}\left(T\right)=\frac{\left(1+\tau\right)\alpha_{i}+\lambda \frac{dV^{\mathrm{tipped}}\left(T\right)}{dT}}{\left(1+\tau\right)\beta_{i}}
\end{equation}

The derivation of the Markov strategy of region $j$ outside the coalition is similar.

By substituting the optimal strategies for member and non-member regions into the HJB equations, we obtain:  

\begin{equation}
\begin{aligned}
rV^{\mathrm{tipped}}\left(T\right) = 
&\sum_{i\in \mathcal{S}} \left[\frac{(1+\tau)\alpha_i^2}{2\beta_i}-\left(1+\ell\right) \left(\frac{1}{2}\rho_{i}T(t)^{2}+\eta_{i}T(t)\right)\right]  \\
& +\sum_{i\in \mathcal{S}}\frac{\lambda^2}{2(1+\tau)\beta_{i}}(\frac{d V^{\mathrm{tipped}}(T)}{d T})^2+\lambda A\frac{dV^{\mathrm{tipped}}(T)}{dT}\\
 & +\frac{dV^{\mathrm{tipped}}(T)}{dT}\sum_{j\notin \mathcal{S}}\frac{\lambda^2}{\beta_{j}}\frac{dW_j^{\mathrm{tipped}}(T)}{dT}
\end{aligned}
\end{equation}

For each non-member region $j\notin \mathcal{S}$, the substituted HJB equation is:
\begin{equation}
\begin{aligned}
rW^{\mathrm{tipped}}_j\left(T\right) = 
& \frac{\alpha_j^2}{2\beta_j}-\left(1+\ell\right)\left(\frac{1}{2}\rho_jT(t)^2+\eta_jT(t)\right) \\
& + \frac{\lambda^2}{2\beta_j}\left(\frac{dW_j^{\mathrm{tipped}}(T)}{dT}\right)^2 + \lambda A\frac{dW_j^{\mathrm{tipped}}(T)}{dT} \\
 & +\frac{dW_j^{\mathrm{tipped}}(T)}{dT}  \sum_{\substack{k \notin \mathcal{S} \\ k \neq j}}\frac{\lambda^2}{\beta_{k}}\frac{dW_k^{\mathrm{tipped}}(T)}{dT}\\
 & +\frac{dW_j^{\mathrm{tipped}}(T)}{dT}\sum_{i\in \mathcal{S}}\frac{\lambda^2}{(1+\tau)\beta_{i}}\frac{dV^{\mathrm{tipped}}(T)}{dT}
\end{aligned}
\end{equation}

\noindent where $A\equiv\sum_{i \in \mathcal{S}}\alpha_{i}/\beta_{i}+\sum_{j \notin \mathcal{S}}\alpha_{j}/\beta_{j}$.

Given the linear-quadratic structure of the game, we postulate quadratic forms for the value functions: 
\begin{equation}
\begin{cases}
V^{{\mathrm{tipped}}}\left(T\right)=v_{0}+v_{1}T+\frac{1}{2}v_{2}T^{2}\\
W^{{\mathrm{tipped}}}_{j}\left(T\right)=w_{j,0}+w_{j,1}T+\frac{1}{2}w_{j,2}T^{2}
\end{cases}
\end{equation}
\noindent where $v_{0}$, $v_{1}$, $v_{2}$, $w_{j,0}$, $w_{j,1}$, and $w_{j,2}$ are undetermined coefficients 
\footnote{To simplify our discussion, these undetermined coefficients will be used repeatedly in the following individual sections}.

Substituting these into the HJB equations yields three coupled systems for the coefficients: 

For coalition members:
\begin{equation}
\begin{cases}
\frac{1}{2}rv_{2}=&\sum_{i\in \mathcal{S}}\left(\frac{\lambda^{2}{v_{2}}^{2}}{2\left(1+\tau\right)\beta_{i}}-\frac{1}{2}\left(1+\ell\right)\rho_{i}\right)+v_{2}\lambda^{2}\sum_{j\notin \mathcal{S}}\frac{w_{j,2}}{\beta_{j}}\\
rv_{1}=&\sum_{i\in \mathcal{S}}\left(\frac{\lambda^{2}v_{1}v_{2}}{\left(1+\tau\right)\beta_{i}}-\left(1+\ell\right)\eta_{i}\right)+\lambda Av_{2}+\lambda^{2}\sum_{j\notin \mathcal{S}}\frac{v_{1}w_{j,2}+v_{2}w_{j,1}}{\beta_{j}}\\
rv_{0}=&\sum_{i\in \mathcal{S}}\left(\frac{\left(1+\tau\right)^{2}{\alpha_{i}}^{2}+\lambda^{2}{v_{1}}^{2}}{2\left(1+\tau\right)\beta_{i}}\right)+ v_{1}\left(\lambda A+\lambda^{2}\sum_{j\notin \mathcal{S}}\frac{w_{j,1}}{\beta_{j}}\right)
\end{cases}\label{eq:linear_quad_coalition}
\end{equation}

For non-members ($\forall j\notin \mathcal{S}$):
\begin{equation}
\begin{cases}
\frac{1}{2}rw_{j,2}=
\frac{\lambda^{2}{w_{j,2}}^{2}}{2\beta_{j}}
-\frac{1}{2}\left(1+\ell\right)\rho_{j}
+w_{j,2}\lambda^{2}
\left(\sum_{i\in \mathcal{S}}\frac{v_{2}}{\left(1+\tau\right)\beta_{i}}+\sum_{\substack{k\notin\mathcal{S} \\ k \neq j}}\frac{w_{k,2}}{\beta_{k}}\right)\\[1ex]
\begin{split}
\textstyle
rw_{j,1}={}&\frac{\lambda^{2}w_{j,1}w_{j,2}}{\beta_{j}}-\left(1+\ell\right)\eta_{j}+\lambda Aw_{j,2}\\
&\textstyle+\lambda^{2}\left(\sum_{i\in \mathcal{S}}\frac{w_{j,1}v_{2}+w_{j,2}v_{1}}{\left(1+\tau\right)\beta_{i}}+w_{j,1}\sum_{\substack{k \notin \mathcal{S} \\ k \neq j}}\frac{w_{k,2}}{\beta_{k}}+w_{j,2}\sum_{\substack{k \notin \mathcal{S} \\ k \neq j}}\frac{w_{k,1}}{\beta_{k}}\right)\end{split}
\\[1ex]
\begin{split}
\textstyle
rw_{j,0}={}
&\frac{{\alpha_{j}}^{2}+\lambda^{2}{w_{j,1}}^{2}}{2\beta_{j}}\\
&\textstyle+w_{j,1}\left(\lambda A+\lambda^{2}\left(\sum_{i\in \mathcal{S}}\frac{v_{1}}{\left(1+\tau\right)\beta_{i}}+\sum_{\substack{k \notin \mathcal{S} \\ k \neq j}}\frac{w_{k,1}}{\beta_{k}}\right)\right)
\end{split}
\end{cases}
\label{eq:linear_quad_noncoalition}
\end{equation}

The quadratic equations for $v_2$ and $w_{j,2}$ admit two algebraic branches. Since the coefficients are jointly determined through
\[
u_2=\sum_{i\in\mathcal S}\frac{v_2}{(1+\tau)\beta_i}
+\sum_{j\notin\mathcal S}\frac{w_{j,2}}{\beta_j},\]
the admissible solution must satisfy the entire coupled Riccati system rather than the individual quadratic equations alone. We restrict attention to real-valued solutions with concave value functions and stable closed-loop temperature dynamics. In particular, the quadratic value functions imply that optimal emissions are linear in temperature, then the transition law of temperature can be rewritten as 
\[\dot T=\lambda A+\lambda^2u_1+\lambda^2u_2T,\]
so asymptotic stability of the closed-loop temperature dynamics requires $u_2<0$. Among the solutions to the coupled system under our parameterization, the admissible solution satisfying these conditions corresponds to the positive square-root branch reported in Table \ref{tab.A1}.

\begin{table}[h!]
\small
\centering
\QTR{caption}{Values of undetermined coefficients}
\label{tab.A1}
\begin{tabular}{ll}
\hline
{Member regions} & {Non-member regions}\\
\hline 
{$v_{2}=\frac{\mathrel{\lambda^{2}u_{2}-\frac{1}{2}r+\sqrt{\left(\lambda^{2}u_{2}-\frac{1}{2}r\right)^{2}-\left(\sum_{i\in \mathcal{S}}\frac{\lambda^{2}}{\left(1+\tau\right)\beta_{i}}\right)\left((1+\ell)\sum_{i\in \mathcal{S}}\rho_{i}\right)}}}{\sum_{i\in \mathcal{S}}\frac{\lambda^{2}}{\left(1+\tau\right)\beta_{i}}}$} & {$w_{j,2}=\frac{\mathrel{\lambda^{2}u_{2}-\frac{1}{2}r+\sqrt{\left(\lambda^{2}u_{2}-\frac{1}{2}r\right)^{2}-(1+\ell)\frac{\lambda^{2}}{\beta_{j}}\rho_{j}}}}{\frac{\lambda^{2}}{\beta_{j}}}$}\\
\hline 
{$v_{1}=\frac{\left(\lambda^{2}u_{1}+\lambda A\right)v_{2}-(1+\ell)\sum_{i\in \mathcal{S}}\eta_{i}}{\left(\sum_{i\in \mathcal{S}}\frac{\lambda^{2}v_{2}}{\left(1+\tau\right)\beta_{i}}\right)-\left(\lambda^{2}u_{2}-r\right)}$} & {$w_{j,1}=\frac{\left(\lambda^{2}u_{1}+\lambda A\right)w_{j,2}-(1+\ell)\eta_{j}}{\frac{\lambda^{2}w_{j,2}}{\beta_{j}}-\left(\lambda^{2}u_{2}-r\right)}$}\\
\hline 
{$v_{0}=\frac{\sum_{i\in\mathcal{S}}\frac{(1+\tau)\alpha_i^2}{2\beta_i}-\frac{\lambda^2 v_1^2}{2}\sum_{i\in\mathcal{S}}\frac{1}{(1+\tau)\beta_i}+(\lambda^2 u_1+\lambda A)v_1}{r}$} & 
{$w_{j,0}=\frac{\frac{\alpha_j^2+\lambda^2 w_{j,1}^2}{2\beta_j}+(\lambda^2 u_1+\lambda A)w_{j,1}-\frac{\lambda^2 w_{j,1}^2}{\beta_j}}{r}$}\\
\hline
\end{tabular}
\centering
\end{table}

\noindent where 
$$u_{1}=\frac{\left(\sum_{i\in \mathcal{S}}\frac{1}{\left(1+\tau\right)\beta_{i}}\right)\left[\frac{\lambda Av_{2}-(1+\ell)\sum_{i\in \mathcal{S}}\eta_{i}}{\left(\sum_{i\in \mathcal{S}}\frac{\lambda^{2}v_{2}}{\left(1+\tau\right)\beta_{i}}\right)-\left(\lambda^{2}u_{2}-r\right)}\right]+\left\{ \sum_{j\notin \mathcal{S}}\frac{\lambda Aw_{j,2}-(1+\ell)\eta_{j}}{\beta_{j}\left[\frac{\lambda^{2}w_{j,2}}{\beta_{j}}-\left(\lambda^{2}u_{2}-r\right)\right]}\right\} }{1-\left[\frac{\sum_{i\in \mathcal{S}}\frac{\lambda^{2}v_{2}}{\left(1+\tau\right)\beta_{i}}}{\left(\sum_{i\in \mathcal{S}}\frac{\lambda^{2}v_{2}}{\left(1+\tau\right)\beta_{i}}\right)-\left(\lambda^{2}u_{2}-r\right)}+\sum_{j\notin \mathcal{S}}\frac{\frac{\lambda^{2}w_{j,2}}{\beta_{j}}}{\frac{\lambda^{2}w_{j,2}}{\beta_{j}}-\left(\lambda^{2}u_{2}-r\right)}\right]}$$,

and $u_{2}$ is the solution to the following equation: 
\begin{equation}
\label{Eq.u2}
\begin{aligned}
0=&\left (n-m\right )\lambda^{2}u_{2}\\
&-\frac{1}{2}r\left(n-m+1\right)+\sum_{j\notin \mathcal{S}}\mathrel{\sqrt{\left(\lambda^{2}u_{2}-\frac{1}{2}r\right)^{2}-(1+\ell)\frac{\lambda^{2}}{\beta_{j}}\rho_{j}}}\\
&+\mathrel{\sqrt{\left(\lambda^{2}u_{2}-\frac{1}{2}r\right)^{2}-\left(\sum_{i\in \mathcal{S}}\frac{\lambda^{2}}{\left(1+\tau\right)\beta_{i}}\right)\left((1+\ell)\sum_{i\in \mathcal{S}}\rho_{i}\right)}}\\ 
\end{aligned}
\end{equation}

For a coalition $\mathcal{S}$ with $m=\left|\mathcal{S}\right|$ member regions, by substituting the solved coefficients into the optimal emission strategies for both member and non-member regions, we obtain the complete equilibrium solutions, including the explicit form of the value functions. 

Since the value functions are quadratic in the temperature state $T$, their first-order derivatives with respect to $T$ are linear functions of $T$. From the first-order conditions, the optimal emission strategies depend on the marginal value of the temperature state. Therefore, the optimal emission strategies in both the no-tipping and post-tipping regimes are affine functions of temperature. 
This establishes that the equilibrium emission strategies are linear in temperature, as stated in Proposition~\ref{prop1}.

\subsection{Existence and Uniqueness of Analytical Solutions}
\label{appA2}
Based on the results in Appendix \ref{appA1}, once the admissible algebraic branch is identified, the existence and uniqueness of the quadratic coefficients $v_2$ and $w_{j,2}$ can be reduced to the existence and uniqueness of the solution for $u_2$ in Equation \ref{Eq.u2}:
\begin{equation*}
\begin{aligned}
\frac{1}{2}r &= (n-m)\left(\lambda^{2}u_{2}-\frac{1}{2}r\right) \\
&\quad + \sqrt{\left(\lambda^{2}u_{2}-\frac{1}{2}r\right)^{2}-\left(\sum_{i\in \mathcal{S}}\frac{\lambda^{2}}{(1+\tau)\beta_{i}}\right)\left((1+\ell)\sum_{i\in \mathcal{S}}\rho_{i}\right)} \\
&\quad + \sum_{j\notin \mathcal{S}}\sqrt{\left(\lambda^{2}u_{2}-\frac{1}{2}r\right)^{2}-(1+\ell)\frac{\lambda^{2}}{\beta_{j}}\rho_{j}}
\end{aligned}
\end{equation*}

Let $x = \lambda^{2}u_{2}-{r}/{2}$, and define $G(x)$ as:
\begin{equation}
\begin{aligned}
G(x) = &(n-m)x + \sum_{j\notin \mathcal{S}}\sqrt{x^{2}-(1+\ell)\frac{\lambda^{2}}{\beta_{j}}\rho_{j}} \\
& + \sqrt{x^2-\left(\sum_{i\in \mathcal{S}}\frac{\lambda^{2}}{(1+\tau)\beta_{i}}\right)\left((1+\ell)\sum_{i\in \mathcal{S}}\rho_{i}\right)}
\end{aligned}
\end{equation}

Define 
\[ d= \sqrt{\max\left\{ \left(\sum_{i\in \mathcal{S}}\frac{\lambda^{2}}{(1+\tau)\beta_{i}}\right)\left((1+\ell)\sum_{i\in \mathcal{S}}\rho_{i}\right), \max_{j\notin \mathcal{S}}\left((1+\ell)\frac{\lambda^{2}}{\beta_{j}}\rho_{j}\right) \right\}}\]

For the grand coalition, where there are no non-members, we define
\[d=\sqrt{\left(\sum_{i\in\mathcal S}\frac{\lambda^2}{(1+\tau)\beta_i}\right)\left((1+\ell)\sum_{i\in\mathcal S}\rho_i\right)}.
\]

It is evident that the domain of $G(x)$ is $(-\infty,-d] \cup [d,+\infty)$. Consequently, the original problem is further transformed into finding the intersection of the curve $y=G(x)$ and the horizontal line $y={r}/{2}$ over this domain.

It remains to verify that the positive square-root branch used in Equation \ref{Eq.u2} is the only branch consistent with the definition of $u_2$. Let $s_{\mathcal S}\in\{-1,+1\}$ and $s_j\in\{-1,+1\}$ denote the algebraic branches for $v_2$ and $w_{j,2}$, respectively. The general consistency condition can be written as
\begin{equation}
\frac{r}{2}=(n-m)x+s_{\mathcal S}\sqrt{x^2-\left(\sum_{i\in\mathcal S}\frac{\lambda^2}{(1+\tau)\beta_i}\right)\left((1+\ell)\sum_{i\in\mathcal S}\rho_i\right)}+\sum_{j\notin\mathcal S}s_j\sqrt{x^2-(1+\ell)\frac{\lambda^2}{\beta_j}\rho_j}.
\end{equation}

For $x\leq-d$, let $x=-M$, where $M\geq d$. If at least one negative square-root branch is selected, at most $n-m$ of the $n-m+1$ square-root terms carry a positive sign. Since every square-root term is no greater than $M$, the right-hand side is no greater than
\[-(n-m)M+(n-m)M=0,\]
which contradicts $r/2>0$. Therefore, all square-root terms must take the positive branch. This establishes the branch used in Equation \ref{Eq.u2}.

For $x\in[d,+\infty)$, $G(x)$ is strictly increasing. However, this branch cannot yield an economically admissible solution. Since $x\ge d>0$ and $0\leq\sqrt{x^2-a}<x$ for any $a>0$, both algebraic roots of $v_2$ and $w_{j,2}$ are positive. Hence, the corresponding value functions are convex rather than concave in temperature. We therefore exclude the interval $[d,+\infty)$ and restrict attention to $x\in(-\infty,-d]$.

For $x \in (-\infty,-d)$, taking the derivative of $G(x)$ with respect to $x$ yields:
\begin{equation}
\begin{aligned}
\frac{dG(x)}{dx} = &(n-m)- \sum_{j\notin \mathcal{S}}\frac{1}{\sqrt{1-(1+\ell)\frac{\lambda^{2}}{\beta_{j}}\rho_{j}\frac{1}{x^2}}}  \\
&- \frac{1}{\sqrt{1-\left(\sum_{i\in \mathcal{S}}\frac{\lambda^{2}}{(1+\tau)\beta_{i}}\right)\left((1+\ell)\sum_{i\in \mathcal{S}}\rho_{i}\right)\frac{1}{x^2}}}
\end{aligned}
\end{equation}

The first term, $n-m$, is a non-negative constant. For the remaining terms, as $x$ increases from $-\infty$ toward $-d$, the term $1/x^2$ increases, causing the denominator $\sqrt{\dots}$ to decrease. Consequently, the negative fraction $-{1}/{\sqrt{\dots}}$ is monotonically decreasing with respect to $x$. Hence, ${dG(x)}/{dx}$ is strictly decreasing on $(-\infty, -d)$. 

Furthermore, taking the limit as $x \to -\infty$:
\begin{equation}
\begin{aligned}
\lim_{x \to -\infty} \frac{dG(x)}{dx} &= (n-m) -
\lim_{x \to -\infty} \left[ 
\begin{split}
&\frac{1}{\sqrt{1-\left(\sum_{i\in \mathcal{S}}\frac{\lambda^{2}}{(1+\tau)\beta_{i}}\right)\left((1+\ell)\sum_{i\in \mathcal{S}}\rho_{i}\right)\frac{1}{x^2}}} \\
&+\sum_{j\notin \mathcal{S}}\frac{1}{\sqrt{1-(1+\ell)\frac{\lambda^{2}}{\beta_{j}}\rho_{j}\frac{1}{x^2}}}
\end{split} \right] \\
&= (n-m) - 1 - \sum_{j\notin \mathcal{S}} 1 \\
&= (n-m) - 1 - (n-m) \\
&= -1 < 0
\end{aligned}
\end{equation}

Given $\lim_{x \to -\infty} {dG(x)}/{dx} = -1 < 0$ and ${dG(x)}/{dx}$ is monotonically decreasing, we have ${dG(x)}/{dx} < 0$ for all $x \in (-\infty, -d)$. Hence $G(x)$ is strictly decreasing on that interval. Moreover, as $x\to-\infty$, each square-root term satisfies
$\sqrt{x^2-a}\sim -x$.
Since there are $n-m+1$ square-root terms in $G(x)$, we obtain
\begin{equation}
\lim_{x\to-\infty}G(x)
=
\lim_{x\to-\infty}
\left[(n-m)x+(n-m+1)(-x)\right]
=
+\infty.
\end{equation}

We observe a mathematical boundary at $x = -d$. While recent literature highlights that differential games may involve “jump points” tied to discontinuous Markovian strategies \citep{jaakkola2026markov}, pinpointing the exact locations and values of such jump points in heterogeneous models proves exceptionally intractable. Therefore, We assume a continuously differentiable value function to preclude intractable discontinuous Markovian jump strategies. Under this premise, the singularity at $x = -d$ acts simply as a continuous boundary rather than a strategic jump point. 

Therefore, the admissible domain is restricted to $x\in(-\infty,-d]$, which, from $x=\lambda^2u_2-r/2$, corresponds to
\begin{equation*}
u_2\leq\frac{r-2d}{2\lambda^2}.
\end{equation*}

Since $G(x)$ is strictly decreasing on $(-\infty,-d]$, its minimum on this interval is attained at $x=-d$. By the definition of $d$, at least one of the $n-m+1$ square-root terms in $G(-d)$ is equal to zero, while each of the remaining $n-m$ terms is no greater than $d$. Therefore,
\begin{equation}
\begin{aligned}
G(-d)
&=-(n-m)d
+\sum_{j\notin\mathcal S}
\sqrt{d^2-(1+\ell)\frac{\lambda^2}{\beta_j}\rho_j}\\
&\quad+
\sqrt{
d^2-
\left(\sum_{i\in\mathcal S}
\frac{\lambda^2}{(1+\tau)\beta_i}\right)
\left((1+\ell)\sum_{i\in\mathcal S}\rho_i\right)
}\leq 0.
\end{aligned}
\end{equation}

Since $r>0$, we have
\[
G(-d)\leq0<\frac{r}{2}.
\]

Combined with $\lim_{x\to-\infty}G(x)=+\infty$ and the strict monotonicity and continuity of $G(x)$ on $(-\infty,-d]$, the intermediate value theorem implies that there exists a unique $x^*\in(-\infty,-d)$ satisfying
\[
G(x^*)=\frac{r}{2}.
\]

Hence, Equation \ref{Eq.u2} admits a unique solution for $u_2$ in the admissible domain.

Moreover, under the positive square-root branch, $x^*+\sqrt{(x^*)^2-a}<0$, implying $v_2<0$ and $w_{j,2}<0$. Hence, the corresponding value functions are strictly concave in temperature.

Finally, once the unique $u_2$ is determined, the coefficients $v_2$ and $w_{j,2}$ are uniquely determined by the positive square-root branches in Table \ref{tab.A1}. Moreover, the denominators in the expressions for $v_1$ and $w_{j,1}$ satisfy
\[
\sum_{i\in\mathcal S}
\frac{\lambda^2v_2}{(1+\tau)\beta_i}
-(\lambda^2u_2-r)=\frac{r}{2}+\sqrt{x^{*2}-
\left(\sum_{i\in\mathcal S}\frac{\lambda^2}{(1+\tau)\beta_i}\right)\left((1+\ell)\sum_{i\in\mathcal S}\rho_i\right)}>0,\]
and
\[
\frac{\lambda^2w_{j,2}}{\beta_j}-(\lambda^2u_2-r)=\frac{r}{2}+\sqrt{x^{*2}-(1+\ell)\frac{\lambda^2}{\beta_j}\rho_j}>0.
\]
Hence, $u_1$, $v_1$, and $w_{j,1}$ are uniquely determined. Given these coefficients and $r>0$, the constant terms $v_0$ and $w_{j,0}$ are also uniquely determined. Therefore, the admissible quadratic solution is unique.

\subsection{Proof of Lemma \ref{Lemma1}}\label{appA3}
Consider member regions maximizing their expected payoffs collectively: 
\begin{equation}
\begin{aligned}
\max_{q_{i}(t),\, {i \in \mathcal{S}}}
\mathbb{E}
\sum_{i\in \mathcal{S}}
\Bigg\{
    & \int_{0}^{t_{1}} e^{-rt}
    \left[
        \left(1+\tau\right)\pi_i\left(q_i(t)\right)
        -D_i\left(T(t)\right)
    \right]dt
    \\
    & +\int_{t_{1}}^{+\infty} e^{-rt}
    \left[
        \left(1+\tau\right)\pi_i\left(q_i(t)\right)
        -(1+\ell)D_i\left(T(t)\right)
    \right]dt
\Bigg\}.
\end{aligned}
\end{equation}

The state variable evolves as ${dT(t)}/{dt} = \lambda\left(\sum_{i\in \mathcal{S}}q_i(t) + \sum_{j\notin \mathcal{S}}q_j(t)\right)$, $T(0) = T_0$. With hazard rate $H(T) = \chi \max(0, T-\underline{T})$, the value function ${V^{\mathrm{pretip}}}\left(T\right)$ satisfies: 
\begin{equation}
\begin{aligned}
rV^{\mathrm{pretip}}\left(T\right)
=\max_{q_i(t),\, i\in\mathcal{S}}\Bigg[
&\sum_{i\in\mathcal{S}}(1+\tau)\left(\alpha_iq_i(t)-\frac{1}{2}\beta_iq_i(t)^2\right)\\
&-\sum_{i\in\mathcal{S}}\left(\frac{1}{2}\rho_iT(t)^2+\eta_iT(t)\right)\\
&+\lambda\frac{dV^{\mathrm{pretip}}\left(T\right)}{dT}\left(\sum_{i\in\mathcal{S}}q_i(t)+\sum_{j\notin\mathcal{S}}q_j(t)\right)\\
&+H\left(T\right)\left(V^{\mathrm{tipped}}\left(T\right)-V^{\mathrm{pretip}}\left(T\right)\right)
\Bigg].
\end{aligned}
\end{equation}

The first-order condition yields:
\begin{equation}
q_{i}^{\mathrm{pretip}*}\left(T\right)=\frac{\left(1+\tau\right)\alpha_{i}+\lambda \frac{dV^{\mathrm{pretip}}\left(T\right)}{dT}}{\left(1+\tau\right)\beta_{i}}
\end{equation}

The derivation of the Markov strategy of region $j$ outside the coalition is similar.

When $T<\underline T$, the hazard rate satisfies $H(T)=0$. 
In this regime, the pre-tipping HJB system retains the same linear-quadratic structure as the post-tipping system, except that climate damages are not amplified by the tipping event.
Therefore, the pre-tipping value functions can be represented by quadratic forms over the regime $T<\underline T$.

For temperature levels $T\ge \underline{T}$, substituting the optimal emission strategies $q_{i}^{\mathrm{pretip}*}$ and $q_{j}^{\mathrm{pretip}*}$ into the HJB equations, we can obtain:
\begin{equation}
\begin{aligned}
    (r+H(T))V^{\mathrm{pretip}}\left(T\right) = 
&\sum_{i\in \mathcal{S}} \left( \frac{(1+\tau)\alpha_i^2}{2\beta_i}-\left(\frac{1}{2}\rho_iT(t)^2+\eta_i T(t)\right)\right) \\
& +\frac{\lambda^2}{2(1+\tau)\beta_{I}}(\frac{d V^{\mathrm{pretip}}(T)}{d T})^2+\lambda A\frac{dV^{\mathrm{pretip}}(T)}{dT}\\
 & +\frac{dV^{\mathrm{pretip}}(T)}{dT}\sum_{j\notin \mathcal{S}}\frac{\lambda^2}{\beta_{j}}\frac{dW_j^{\mathrm{pretip}}(T)}{dT}+H(T)V^{\mathrm{tipped}}(T)
\end{aligned}
\end{equation}

\begin{equation}
\begin{aligned}
\left(r+H(T)\right)W_j^{\mathrm{pretip}}(T)
= & \frac{\alpha_j^2}{2\beta_j} -\left(\frac{1}{2}\rho_jT^2+\eta_jT \right)\\
& + \frac{\lambda^2}{2\beta_j} \left(\frac{dW_j^{\mathrm{pretip}}(T)}{dT}\right)^2 +\lambda A \frac{dW_j^{\mathrm{pretip}}(T)}{dT}\\
&+\frac{dW_j^{\mathrm{pretip}}(T)}{dT}\sum_{\substack{k\notin\mathcal S\\k\neq j}}
\frac{\lambda^2}{\beta_k}\frac{dW_k^{\mathrm{pretip}}(T)}{dT} \\
&+ \frac{dW_j^{\mathrm{pretip}}(T)}{dT}\sum_{i\in\mathcal S}
\frac{\lambda^2}{(1+\tau)\beta_i}V_T^{\mathrm{pretip}}(T)\\
&+H(T)W_j^{\mathrm{tipped}}(T).
\end{aligned}
\end{equation}
where ${1}/{\beta_{I}} = \sum_{j\in \mathcal{S}} {1}/{\beta_i}$ and ${1}/{\beta_{J}} = \sum_{j\notin \mathcal{S}} {1}/{\beta_j}$.

For $T > \underline{T}$, we analyze the asymptotic behavior of the HJB equations using the method of dominant balance. Suppose the value functions $V^{\mathrm{pretip}}$ and $W_j^{\mathrm{pretip}}$ are polynomials of degree $n$ (i.e., $\mathcal{O}(T^n)$), with non-zero leading coefficients $v_n$ and $w_n$, respectively. 

On the left-hand side, the multiplication by the linear hazard rate $H(T) = \chi(T-\underline{T})$ results in a highest degree of $\mathcal{O}(T^{n+1})$. On the right-hand side, the highest-degree terms arise from three sources: the quadratic damage costs which are $\mathcal{O}(T^2)$, the hazard term $H(T)V^{\mathrm{tipped}}(T)$ which is $\mathcal{O}(T^3)$ (since $V^{\mathrm{tipped}}$ is quadratic), and the non-linear derivative combinations $({dV}/{dT})^2$, $({dW_j}/{dT})^2$ and $({dV}/{dT})({dW_j}/{dT})$, which are $\mathcal{O}(T^{2n-2})$. Thus, the highest degree on the right-hand side is exactly $\max(\mathcal{O}(T^{2n-2}), \mathcal{O}(T^3))$.

For the highest-degree terms to balance asymptotically between both sides, we must have $n+1 = \max(2n-2, 3)$. 

If $n\ge 4$, the apparent highest degree on the right-hand side is
$2n-2$, which strictly exceeds the left-hand side degree $n+1$.
Hence, an exact finite-degree polynomial solution would require the
coefficients of all $\mathcal O(T^{2n-2})$ terms to vanish simultaneously.

Let
$$ V^{\mathrm{pretip}}(T)\sim v_nT^n,
\qquad
W_j^{\mathrm{pretip}}(T)\sim w_{j,n}T^n,
\quad j\notin\mathcal S,
$$
where $v_n$ and $w_{j,n}$ denote the leading coefficients. The
$\mathcal O(T^{2n-2})$ coefficients in the $V^{\mathrm{pretip}}$ equation and in the $W_j^{\mathrm{pretip}}$ equation are respectively proportional to

\begin{equation}
C_V = \frac{\lambda^2}{2} \left(\sum_{i\in\mathcal S} \frac{1}{(1+\tau)\beta_i}\right) v_n^2 + \lambda^2 v_n \sum_{j\notin\mathcal S} \frac{w_{j,n}}{\beta_j}
\end{equation} 
\begin{equation}
C_{W_j} = \frac{\lambda^2}{2\beta_j} w_{j,n}^2+ \lambda^2 w_{j,n} \left[\sum_{i\in\mathcal S} \frac{v_n}{(1+\tau)\beta_i} + \sum_{\substack{k\notin\mathcal S \\ k\neq j}} \frac{w_{k,n}}{\beta_k}
\right]
\end{equation}

Therefore, cancellation of the highest-order derivative terms would require $C_V=0,C_{W_j}=0,\forall j\notin\mathcal S$.

Assume first that $v_n\neq 0$ and $w_{j,n}\neq 0$ for at least one non-member region $j\notin\mathcal S$. 
From $C_V=0$, we obtain
\begin{equation}
\sum_{j\notin\mathcal S}\frac{w_{j,n}}{\beta_j}=-\frac{1}{2}\sum_{i\in\mathcal S}\frac{v_n}{(1+\tau)\beta_i}.
\end{equation}

Substituting this condition into $C_{W_j}=0$ yields, for each non-member region with $w_{j,n}\neq 0$,
\begin{equation}
\frac{w_{j,n}}{\beta_j}=\sum_{i\in\mathcal S}\frac{v_n}{(1+\tau)\beta_i}.
\end{equation}

Equivalently,
\begin{equation}
w_{j,n} = \beta_j \sum_{i\in\mathcal S} \frac{v_n}{(1+\tau)\beta_i}.
\end{equation}

Since $\beta_j>0$ and $(1+\tau)\beta_i>0$, this implies that $w_{j,n}$ has the same sign as $v_n$ for every such non-member region. Consequently,
$\sum_{j\notin\mathcal S} {w_{j,n}}/{\beta_j}$
must have the same sign as $v_n$, whereas the condition implied by $C_V=0$ requires it to have the opposite sign of $v_n$. This contradiction implies that the highest-order coefficients cannot vanish simultaneously unless
$$
v_n=0,
\qquad
w_{j,n}=0,
\quad \forall j\notin\mathcal S
$$

However, this contradicts the assumption that $v_n$ and $w_{j,n}$ are leading coefficients of degree-$n$ polynomial value functions. Hence, the $\mathcal O(T^{2n-2})$ terms cannot be eliminated simultaneously.

Consequently, for $n\ge4$, the right-hand side contains a non-vanishing $\mathcal O(T^{2n-2})$ term in at least one HJB equation, whereas the left-hand side has degree only $n+1$. Since $2n-2>n+1 \quad\text{for all } n\ge4$,
the polynomial ansatz cannot generally close at degree $n$. Thus, an exact finite-degree polynomial solution of degree $n\ge4$ is generically ruled out in the heterogeneous model.

Even for cases where $n < 4$, an exact polynomial solution does not exist. 

If $n=2$, the highest-degree balance condition appears to hold at $\mathcal O(T^3)$, because the linear hazard rate $H(T)$ multiplied by a quadratic value function generates cubic terms. However, substituting a general quadratic form for $V^{\mathrm{pretip}}$ and $W_j^{\mathrm{pretip}}$ into the coupled HJB equations generates matching conditions for four powers, namely $T^3,T^2,T^1$, and $T^0$, whereas each quadratic value function contains only three undetermined coefficients. Therefore, the coefficient-matching system is generally overdetermined. Except under restrictive parameter configurations, a quadratic polynomial ansatz cannot provide a valid exact solution.

If $n=3$, the highest-degree balance condition holds structurally at $\mathcal O(T^4)$. Nevertheless, substituting a general cubic polynomial into the HJB equations generates matching conditions for five powers, namely $T^4,T^3,T^2,T^1$, and $T^0$, while each cubic value function contains only four undetermined coefficients. The resulting coefficient system is again overdetermined. Hence, an exact cubic polynomial solution is not generally attainable in the heterogeneous model.

Therefore, the pre-tipping HJB system does not generally admit a tractable finite-degree polynomial representation under the heterogeneous specification. Consequently, we rely on numerical methods to approximate the pre-tipping value functions.

\subsection{Numerical Solution for State ``$pretip$'' via Chebyshev Polynomial Approximation}\label{appA4}
We employ Chebyshev polynomial approximation on $\left[\underline{T},\overline{T}\right]$, where $\underline{T}$ is the tipping threshold and $\overline{T}$ is the upper bound of temperature (in our numerical examples we choose $\overline{T}=4^\circ C$):

\begin{equation}
\hat{V}\left(T\right)=\sum_{k=0}^{K}a_k \phi_k(T)
\end{equation}

\noindent where $a_{k}$ is the Chebyshev coefficients, $\phi_k(T)$ are basis Chebyshev polynomials adapted to $\left[\underline{T},\overline{T}\right]$ for $k=0,1,2,...$ with degree-$K$. 

We choose $K+1$ Chebyshev nodes:$$T_{d}=\left(z_{d}+1\right)\left(\overline{T}-\underline{T}\right)/2+\underline{T}$$
with $z_{d}=-\cos\left({\left(2d-1\right)\pi}/{\left(2\left(K+1 \right)\right)}\right)$ for $d=1,...,K+1$. The Lagrange values $\left(T_{d},v_{d}\right)$: $d=1,...,K+1$ are given (where $v_{d}=V(T_{d})$), and the coefficients $a_{k}$ can be computed (see \citet{judd1998numerical} and \citet{cai2019computational} for more details): 
\begin{equation}
\begin{cases}
a_{0}=\frac{1}{K+1}\sum_{d=1}^{K+1}v_{d}  \\
a_{k}=\frac{2}{K+1}\sum_{d=1}^{K+1}v_{d}\phi_k(T_d),\ \ \ k=1,...,K
\end{cases}
\end{equation}

We then assume:  
\begin{equation}
\begin{cases}
V^{\mathrm{pretip}}\left(T\right)=\sum_{k=0}^{K}a_{k}\phi_{k}\left(T \right)\\
W^{\mathrm{pretip}}_{j}\left(T\right)=\sum_{k=0}^{K}b_{k,j}\phi_{k}\left(T \right)
\end{cases}
\end{equation}
for $T\in[\underline{T},\overline{T}]$, respectively. Thus, we solve 

\begin{equation}
\begin{split}
0=&-\left[r+H\left(T_{d}\right)\right]\sum_{k=0}^{K}a_{k}\phi_{k}(T_{d})\\
&+\sum_{i\in \mathcal{S}}\left[
\frac{\left((1+\tau)\alpha_{i}
+\lambda\sum_{k=0}^{K}a_{k}\phi_{k,T}(T_{d})\right)^2}
{2(1+\tau)\beta_{i}}
-\frac{1}{2}\rho_{i}T_{d}^{2}
-\eta_{i}T_{d}
\right]\\
&+\lambda
\left(\sum_{k=0}^{K}a_{k}\phi_{k,T}(T_{d})\right)
\sum_{j\notin\mathcal{S}}
\frac{
\alpha_j+\lambda\sum_{k=0}^{K}b_{k,j}\phi_{k,T}(T_d)
}
{\beta_j}\\
&+H\left(T_{d}\right)
\left(v_{0}+v_{1}T_d+\frac{1}{2}v_{2}T_d^{2}\right)
\end{split}
\end{equation}

\begin{equation}
\begin{split}
0=&-\left[r+H\left(T_{d}\right)\right]\sum_{k=0}^{K}b_{k,j}\phi_{k}(T_{d})\\
&+\frac{\left(\alpha_{j}
+\lambda\sum_{k=0}^{K}b_{k,j}\phi_{k,T}(T_{d})\right)^2}
{2\beta_{j}}
-\frac{1}{2}\rho_{j}T_{d}^{2}
-\eta_{j}T_{d}\\
&+\lambda
\left(\sum_{k=0}^{K}b_{k,j}\phi_{k,T}(T_{d})\right)
\left[
\sum_{i\in\mathcal S}
\frac{(1+\tau)\alpha_i+\lambda\sum_{k=0}^{K}a_k\phi_{k,T}(T_d)}
{(1+\tau)\beta_i}
+
\sum_{\substack{h\notin\mathcal S\\h\neq j}}
\frac{\alpha_h+\lambda\sum_{k=0}^{K}b_{k,h}\phi_{k,T}(T_d)}
{\beta_h}
\right]\\
&+H\left(T_{d}\right)
\left(w_{j,0}+w_{j,1}T_d+\frac{1}{2}w_{j,2}T_d^{2}\right)
\end{split}
\end{equation}

\noindent where ${\phi_{k,T}}(T_d)$ denotes the first derivative of $\phi_{k}(T_{d})$ at $T_d$.

We first solve the post-tipping value functions analytically and use them as continuation values in the hazard terms of the pre-tipping HJB equations. The resulting system contains $(K+1)(1+n-m)$ equations in the same number of unknown coefficients, $a_k$ and $b_{k,j}$, and is solved numerically using MATLAB's \texttt{fsolve}. Guided by Lemma~\ref{Lemma1} and robustness checks across approximation degrees\footnote{We approximate the value functions using Chebyshev polynomials of degrees 4, 6, and 8. The resulting equilibrium outcomes are very similar across these specifications.}, we set $K=6$, balancing computational efficiency and approximation accuracy.

\subsection{Numerical Solution for State ``$pre1tip$" via Chebyshev Polynomial Approximation}
\label{appA5}

The solution for State ``$pre1tip$" proceeds recursively across tipping states. In State ``$2tipped$", where both tipping events have occurred, the value functions admit the analytical solutions derived in Appendix \ref{appA1}. In ``$pre2tip$", where only the first tipping event has occurred, the remaining uncertainty concerns the second tipping event, and the corresponding value functions can be obtained using the Chebyshev approximation procedure described in Appendix \ref{appA4}. These ``$pre2tip$" value functions then serve as continuation values when solving State ``$pre1tip$", in which neither tipping event has yet occurred.

In State ``$pre1tip$", regions choose emissions while accounting for the possibility of two future tipping events. Let $t_1$ and $t_2$ denote the stochastic arrival times of the first and second tipping events, respectively. Prior to the first tipping event, climate damages are given by $D_i(T)$. Following the first event, damages are amplified by $(1+\ell_1)$, and following the second event, they are further amplified by $(1+\ell_2)$. The coalition therefore solves
\begin{equation}
\begin{aligned}
\max_{q_i(t),\, i \in \mathcal{S}} \mathbb{E}\sum_{i\in \mathcal{S}}\Bigg\{
&\int_{0}^{t_{1}}e^{-rt}\left[(1+\tau)\pi_i(q_i(t))-D_i(T(t))\right]dt\\
&+\int_{t_{1}}^{t_{2}}e^{-rt}\left[(1+\tau)\pi_i(q_i(t))-(1+\ell_1)D_i(T(t))\right]dt\\
&+\int_{t_{2}}^{\infty}e^{-rt}\left[(1+\tau)\pi_i(q_i(t))-(1+\ell_1)(1+\ell_2)D_i(T(t))\right]dt
\Bigg\}.
\end{aligned}
\end{equation}

Each non-member region $j\notin \mathcal{S}$ solves the corresponding individual optimization problem:
\begin{equation}
\begin{aligned}
\max_{q_j(t)} \mathbb{E}\Bigg\{
&\int_{0}^{t_{1}}e^{-rt}\left[\pi_j(q_j(t))-D_j(T(t))\right]dt\\
&+\int_{t_{1}}^{t_{2}}e^{-rt}\left[\pi_j(q_j(t))-(1+\ell_1)D_j(T(t))\right]dt\\
&+\int_{t_{2}}^{\infty}e^{-rt}\left[\pi_j(q_j(t))-(1+\ell_1)(1+\ell_2)D_j(T(t))\right]dt
\Bigg\}.
\end{aligned}
\end{equation}

The associated pre-first-tipping value functions are denoted by
$V^{\mathrm{pre1tip}}(T)$ for the coalition and
$W_j^{\mathrm{pre1tip}}(T)$ for each non-member region. The first tipping event moves the economy from State ``$pre1tip$" to ``$pre2tip$". Hence, the ``$pre2tip$" value functions,
$V^{\mathrm{pre2tip}}(T)$ and $W_j^{\mathrm{pre2tip}}(T)$,
enter the HJB equations as continuation values. The corresponding HJB equations are
\begin{equation}
\begin{aligned}
rV^{\mathrm{pre1tip}}(T)
=\max_{q_i(t),\, i\in\mathcal{S}}\Bigg[
&\sum_{i\in\mathcal{S}}
\left(
(1+\tau)\left(\alpha_iq_i(t)-\frac{1}{2}\beta_iq_i(t)^2\right)
-\frac{1}{2}\rho_iT(t)^2-\eta_iT(t)
\right)\\
&+\frac{dV^{\mathrm{pre1tip}}(T)}{dT}
\lambda\left(
\sum_{i\in\mathcal{S}}q_i(t)
+\sum_{j\notin\mathcal{S}}q_j(t)
\right)\\
&+H_1(T)\left(
V^{\mathrm{pre2tip}}(T)
-V^{\mathrm{pre1tip}}(T)
\right)
\Bigg],
\end{aligned}
\end{equation}
and, for each $j\notin\mathcal S$,
\begin{equation}
\begin{aligned}
rW_j^{\mathrm{pre1tip}}(T)
=\max_{q_j(t)}\Bigg[
&\left(
\alpha_jq_j(t)
-\frac{1}{2}\beta_jq_j(t)^2
-\frac{1}{2}\rho_jT(t)^2
-\eta_jT(t)
\right)\\
&+\frac{dW_j^{\mathrm{pre1tip}}(T)}{dT}
\lambda\left(
\sum_{i\in\mathcal{S}}q_i(t)
+\sum_{h\notin\mathcal{S}}q_h(t)
\right)\\
&+H_1(T)\left(
W_j^{\mathrm{pre2tip}}(T)
-W_j^{\mathrm{pre1tip}}(T)
\right)
\Bigg].
\end{aligned}
\end{equation}

Using the first-order conditions for coalition members and non-members, we substitute the optimal emission strategies into the HJB equations. As in Appendix \ref{appA4}, we approximate the unknown State ``$pre1tip$" value functions by Chebyshev polynomials, with coefficients $\{a_k\}_{k=0}^{K}$ for the coalition value function and $\{b_{k,j}\}_{k=0}^{K}$ for each non-member value function. Evaluating the resulting HJB residuals at the $K+1$ Chebyshev collocation nodes yields
\begin{equation}
\begin{split}
0=&-\left[r+H_1\left(T_{d}\right)\right]\sum_{k=0}^{K}a_{k}\phi_{k}(T_{d})\\
&+\sum_{i\in \mathcal{S}}\left[
\frac{\left((1+\tau)\alpha_{i}
+\lambda\sum_{k=0}^{K}a_{k}\phi_{k,T}(T_{d})\right)^2}
{2(1+\tau)\beta_{i}}
-\frac{1}{2}\rho_{i}T_{d}^{2}
-\eta_{i}T_{d}
\right]\\
&+\lambda
\left(\sum_{k=0}^{K}a_{k}\phi_{k,T}(T_{d})\right)
\sum_{j\notin\mathcal{S}}
\frac{
\alpha_j+\lambda\sum_{k=0}^{K}b_{k,j}\phi_{k,T}(T_d)
}
{\beta_j}+H_1\left(T_{d}\right)V^{\mathrm{pre2tip}}\left(T_{d}\right),
\end{split}
\end{equation}
and, for each $j\notin\mathcal S$,
\begin{equation}
\begin{split}
0=&-\left[r+H_1\left(T_{d}\right)\right]\sum_{k=0}^{K}b_{k,j}\phi_{k}(T_{d})\\
&+\frac{\left(\alpha_{j}
+\lambda\sum_{k=0}^{K}b_{k,j}\phi_{k,T}(T_{d})\right)^2}
{2\beta_{j}}
-\frac{1}{2}\rho_{j}T_{d}^{2}
-\eta_{j}T_{d}\\
&+\lambda
\left(\sum_{k=0}^{K}b_{k,j}\phi_{k,T}(T_{d})\right)
\left[
\sum_{i\in\mathcal S}
\frac{(1+\tau)\alpha_i+\lambda\sum_{k=0}^{K}a_k\phi_{k,T}(T_d)}
{(1+\tau)\beta_i}
\right.\\
&\hspace{4.5cm}\left.
+\sum_{\substack{h\notin\mathcal S\\ h\neq j}}
\frac{\alpha_h+\lambda\sum_{k=0}^{K}b_{k,h}\phi_{k,T}(T_d)}
{\beta_h}
\right]\\
&+H_1\left(T_{d}\right)W_j^{\mathrm{pre2tip}}\left(T_{d}\right).
\end{split}
\end{equation}

Here, $d=1,\ldots,K+1$ indexes the Chebyshev collocation nodes. The continuation values
$V^{\mathrm{pre2tip}}(T)$ and
$W_j^{\mathrm{pre2tip}}(T)$
are computed beforehand using the procedure described in Appendix \ref{appA4} and are treated as known functions in the State ``$pre1tip$" collocation system.

The resulting nonlinear system jointly determines the Chebyshev coefficients for the coalition and all non-member value functions. Once these coefficients are obtained, the corresponding first-order conditions recover the equilibrium emission strategies in State ``$pre1tip$". Following Appendix \ref{appA4}, we set $K=6$, so that each value function is approximated by a degree-6 Chebyshev polynomial.

\section{Appendix B: Proofs of the Main Theorems}
\label{appB}

\global\long\def\thefigure{B\arabic{figure}}%
\setcounter{figure}{0} 
\global\long\def\thetable{B\arabic{table}}%
\setcounter{table}{0} 
\global\long\def\theequation{B\arabic{equation}}%
\setcounter{equation}{0} 

\subsection{Proof of Theorem \ref{theorem.1}}
\label{appB1}

The coalition partition is taken as given: \(m\) regions are coalition members and \(n-m\) regions are non-members. And the assumption of identical payoff functions is used to obtain tractable representative emission rules.  Thus, the result compares behavior across two exogenously assigned groups, rather than deriving the coalition membership decision from regional heterogeneity.

Because the analysis is conducted in the post-tipping regime, the damage multiplier is constant and equal to \(1+\ell\). The quadratic value functions are therefore
\begin{equation}
V(T,m)=v_0(m)+v_1(m)T+\frac{1}{2}v_2(m)T^2,
\end{equation}
and
\begin{equation}
W(T,m)=w_0(m)+w_1(m)T+\frac{1}{2}w_2(m)T^2,
\end{equation}
where \(V(T,m)\) is the aggregate value function of the \(m\)-member coalition and \(W(T,m)\) is the representative value function of a non-member.

When $\tau=0$, the first-order condition for a coalition member is
\begin{equation}
\alpha-\beta q_i+\lambda V_T(T,m)=0,
\end{equation}
which gives the non-member's optimal emissions is
\begin{equation}
q^M(T,m)=\frac{\alpha+\lambda V_T(T,m)}{\beta}.
\end{equation}

Similarly, the first-order condition for a non-member is
\begin{equation}
\alpha-\beta q_j+\lambda W_T(T,m)=0,
\end{equation}
so that
\begin{equation}
q^N(T,m)=\frac{\alpha+\lambda W_T(T,m)}{\beta}.
\end{equation}

Define
\begin{equation}
u_2(m)=\frac{mv_2(m)+(n-m)w_2(m)}{\beta},
\end{equation}
and let
\begin{equation}
\sigma_m=\lambda^2u_2(m)-\frac{1}{2}r.
\end{equation}
The stabilizing roots can then be written as
\begin{equation}
v_2(m)=\frac{\beta}{m\lambda^2}\left[\sigma_m+\sqrt{\sigma_m^2-\frac{m^2\lambda^2(1+\ell)\rho}{\beta}}\right],
\end{equation}
and
\begin{equation}
w_2(m)=\frac{\beta}{\lambda^2}\left[\sigma_m+\sqrt{\sigma_m^2-\frac{\lambda^2(1+\ell)\rho}{\beta}}\right].
\end{equation}
The consistency condition for $\sigma_m$ is
\begin{equation}
\begin{aligned}
0=&(n-m)\sigma_m-\frac{1}{2}r
+\sqrt{\sigma_m^2-\frac{m^2\lambda^2(1+\ell)\rho}{\beta}}\\
&+(n-m)\sqrt{\sigma_m^2-\frac{\lambda^2(1+\ell)\rho}{\beta}}.
\end{aligned}
\end{equation}
On the stabilizing branch, $\sigma_m<0$.

\subsubsection{Member Versus Non-member under a Fixed Coalition of Size $m$}
\label{appB11}

We first compare a representative coalition member and a representative non-member under the same coalition partition. Define
\begin{equation}
q^{N-M}(T,m)=q^N(T,m)-q^M(T,m).
\end{equation}
Using the first-order conditions,
\begin{equation}
q^{N-M}(T,m)=\frac{\lambda}{\beta}\left[w_1(m)-v_1(m)+\left(w_2(m)-v_2(m)\right)T\right],
\end{equation}
so that
\begin{equation}
\frac{\partial q^{N-M}(T,m)}{\partial T}
=\frac{\lambda}{\beta}\left[w_2(m)-v_2(m)\right].
\end{equation}

Using the stabilizing roots, we can write
\begin{equation}
v_2(m)=\frac{\beta}{\lambda^2}\left[\frac{\sigma_m}{m}
+\sqrt{\left(\frac{\sigma_m}{m}\right)^2-\frac{\lambda^2(1+\ell)\rho}{\beta}}\right],
\end{equation}
and
\begin{equation}
w_2(m)=\frac{\beta}{\lambda^2}\left[\sigma_m
+\sqrt{\sigma_m^2-\frac{\lambda^2(1+\ell)\rho}{\beta}}\right].
\end{equation}

Consider the function
\begin{equation}
f(x)=x+\sqrt{x^2-\frac{\lambda^2(1+\ell)\rho}{\beta}}
\end{equation}
on the relevant negative branch. Its derivative is
\begin{equation}
f'(x)=1+\frac{x}{\sqrt{x^2-\frac{\lambda^2(1+\ell)\rho}{\beta}}}<0,
\end{equation}
implying that $f(x)$ is strictly decreasing. Since $\sigma_m<0$ and $m>1$,
\begin{equation}
\frac{\sigma_m}{m}>\sigma_m.
\end{equation}
Monotonicity of $f$ then implies
\begin{equation}
f\left(\frac{\sigma_m}{m}\right)<f(\sigma_m),
\end{equation}
and hence
\begin{equation}
v_2(m)<w_2(m).
\end{equation}
It follows that
\begin{equation}
\frac{\partial q^{N-M}(T,m)}{\partial T}>0.
\end{equation}
Thus, for a fixed coalition of size $m$, the emission gap between a non-member and a coalition member increases with temperature.

\subsubsection{Joining a Coalition Versus Remaining Outside}
\label{appB12}

We next compare the same representative region when it joins the coalition and when it remains outside. If the region joins, the resulting coalition size is \(m\). If it remains outside, the coalition size is \(m-1\). Define
\begin{equation}
q^\psi(T,m)=q^N(T,m-1)-q^M(T,m).
\end{equation}
Using the first-order conditions,
\begin{equation}
q^\psi(T,m)=\frac{\lambda}{\beta}\left[w_1(m-1)-v_1(m)+\left(w_2(m-1)-v_2(m)\right)T\right].
\end{equation}
Therefore,
\begin{equation}
\frac{\partial q^\psi(T,m)}{\partial T}=\frac{\lambda}{\beta}\left[w_2(m-1)-v_2(m)\right].
\end{equation}

From the stabilizing roots,
\begin{equation}
v_2(m)=\frac{\beta}{\lambda^2}f\left(\frac{\sigma_m}{m}\right),
\qquad
w_2(m-1)=\frac{\beta}{\lambda^2}f(\sigma_{m-1}),
\end{equation}
where
\begin{equation}
f(x)=x+\sqrt{x^2-\frac{\lambda^2(1+\ell)\rho}{\beta}}.
\end{equation}
Since \(f(x)\) is strictly decreasing on the stabilizing branch, it is sufficient to show
\begin{equation}
\frac{\sigma_m}{m}>\sigma_{m-1},
\end{equation}
or equivalently,
\begin{equation}
\sigma_m>m\sigma_{m-1}.
\end{equation}

We now prove this inequality. When \(\tau=0\), \(\sigma_m\) satisfies
\begin{equation}
\begin{aligned}
0=&(n-m)\sigma_m-\frac{1}{2}r+\sqrt{\sigma_m^2-\frac{m^2\lambda^2(1+\ell)\rho}{\beta}}+(n-m)\sqrt{\sigma_m^2-\frac{\lambda^2(1+\ell)\rho}{\beta}},
\end{aligned}
\end{equation}
and \(\sigma_{m-1}\) satisfies
\begin{equation}
\begin{aligned}
0=&(n-m+1)\sigma_{m-1}-\frac{1}{2}r+\sqrt{\sigma_{m-1}^2-\frac{(m-1)^2\lambda^2(1+\ell)\rho}{\beta}}+(n-m+1)\sqrt{\sigma_{m-1}^2-\frac{\lambda^2(1+\ell)\rho}{\beta}}.
\end{aligned}
\end{equation}
Let \(F_m(x)\) denote the left-hand side of the equation for \(\sigma_m\). On the stabilizing branch, \(F_m(x)\) is strictly decreasing because
\begin{equation}
F_m'(x)=(n-m)+\frac{x}{\sqrt{x^2-\frac{m^2\lambda^2(1+\ell)\rho}{\beta}}}+(n-m)\frac{x}{\sqrt{x^2-\frac{\lambda^2(1+\ell)\rho}{\beta}}}<0.
\end{equation}
Since \(F_m(\sigma_m)=0\), the inequality \(\sigma_m>m\sigma_{m-1}\) holds if and only if \(F_m(m\sigma_{m-1})>0\).

Using the equation for \(\sigma_{m-1}\), we obtain
\begin{equation}
\begin{aligned}
F_m(m\sigma_{m-1})=&\frac{m-1}{2}r+m\left[-\sigma_{m-1}-\sqrt{\sigma_{m-1}^2-\frac{(m-1)^2\lambda^2(1+\ell)\rho}{\beta}}\right]\\
&+(n-m)\left[\sqrt{m^2\sigma_{m-1}^2-\frac{\lambda^2(1+\ell)\rho}{\beta}}-m\sqrt{\sigma_{m-1}^2-\frac{\lambda^2(1+\ell)\rho}{\beta}}\right].
\end{aligned}
\end{equation}
The first term is strictly positive because \(m>1\) and \(r>0\). The second term is strictly positive because \(\sigma_{m-1}<0\). The third term is nonnegative because \(n>m\) and
\begin{equation}
\sqrt{m^2\sigma_{m-1}^2-\frac{\lambda^2(1+\ell)\rho}{\beta}}>m\sqrt{\sigma_{m-1}^2-\frac{\lambda^2(1+\ell)\rho}{\beta}}.
\end{equation}
Therefore,
\begin{equation}
F_m(m\sigma_{m-1})>0.
\end{equation}
Hence,
\begin{equation}
\sigma_m>m\sigma_{m-1}.
\end{equation}
It follows that
\begin{equation}
v_2(m)<w_2(m-1).
\end{equation}
Therefore,
\begin{equation}
\frac{\partial}{\partial T}\left[q^N(T,m-1)-q^M(T,m)\right]>0.
\end{equation}
Thus, for a representative region, the emission gap between remaining outside the coalition and joining the coalition increases with temperature.

\subsection{Proof of Theorem \ref{theorem.2}}
\label{appB2}

\textbf{Pre-tipping HJB formulation}: When there is only one tipping event, consider State ``$pretip$'', the pre-tipping regime. Under full cooperation, the expected value function for the coalition is given by:
\begin{equation}
\begin{aligned}
V^{\mathrm{pretip}}(T;\bar{\mathcal{S}})
=\max_{q_i(t),\, i\in\bar{\mathcal{S}}}\mathbb{E}\sum_{i\in\bar{\mathcal{S}}}\Bigg\{
&\int_{0}^{t_{1}}e^{-rt}\left[(1+\tau)\pi_i\left(q_i(t)\right)-D_i\left(T(t)\right)\right]dt\\
&+\int_{t_{1}}^{\infty}e^{-rt}\left[(1+\tau)\pi_i\left(q_i(t)\right)-(1+\ell)D_i\left(T(t)\right)\right]dt
\Bigg\}.
\end{aligned}
\end{equation}

The state variable evolves according to ${dT(t)}/{dt} = \lambda\sum_{i\in \bar{\mathcal{S}}}q_i(t)$ with initial condition $T(0) = T_0$. Given the hazard rate $H(T) = \chi \max(0, T-\underline{T})$, the corresponding Hamilton-Jacobi-Bellman (HJB) equation is:
\begin{equation}
\begin{aligned}
rV^{\mathrm{pretip}}(T;\bar{\mathcal{S}})
=\max_{q_i(t),\, i\in\bar{\mathcal{S}}}\Bigg[
&\sum_{i\in\bar{\mathcal{S}}}\left((1+\tau)\left(\alpha_iq_i(t)-\frac{1}{2}\beta_iq_i(t)^2\right)-\frac{1}{2}\rho_iT(t)^2-\eta_iT(t)\right)\\
&+\lambda\frac{dV^{\mathrm{pretip}}\left(T;\bar{\mathcal{S}}\right)}{dT}\sum_{i\in\bar{\mathcal{S}}}q_i(t)\\
&+H\left(T\right)\left(V^{\mathrm{tipped}}\left(T;\bar{\mathcal{S}}\right)-V^{\mathrm{pretip}}\left(T;\bar{\mathcal{S}}\right)\right)
\Bigg].
\end{aligned}
\end{equation}

Taking the first-order condition yields the optimal emissions strategy:
\begin{equation}
q_{i}^{\mathrm{pretip}*}\left(T;\bar{\mathcal{S}}\right)=\frac{\left(1+\tau\right)\alpha_{i}+\lambda \frac{dV^{\mathrm{pretip}}\left(T;\bar{\mathcal{S}}\right)}{dT}}{\left(1+\tau\right)\beta_{i}}
\end{equation}

By definition, $H(T)=0$ when $T<\underline{T}$. In this region, the pre-tipping HJB reduces to a standard linear-quadratic form without the hazard-jump term. However, it is not identical to the post-tipping HJB, since the post-tipping regime incorporates the amplified damage term $(1+\ell)D_i(T)$. For temperature levels $T\ge \underline{T}$, substituting the optimal emission strategies $q_{i}^{\mathrm{pretip}*}$ into the HJB equation simplifies the expression to:
\begin{equation}
\begin{aligned}
rV^{\mathrm{pretip}}(T;\bar{\mathcal{S}}) = &  \left(1+\tau\right)\sum_{i\in \bar{\mathcal{S}}} \left(\frac{\beta_{i}}{2}(q_i^{\mathrm{pretip*}})^2 \right) - \sum_{i\in \bar{\mathcal{S}}} \left(\frac{1}{2}\rho_{i}T^{2}+\eta_{i}T\right) \\
& + H(T)\left(V^{\mathrm{tipped}}\left(T;\bar{\mathcal{S}}\right)-V^{\mathrm{pretip}}\left(T;\bar{\mathcal{S}}\right)\right)
\end{aligned}
\end{equation}

\textbf{Emission differences between State ``$pretip$'' and State ``$tipped$''}: To analyze the difference in optimal emissions, we consider optimal emissions strategies and the corresponding HJB equation for State ``$tipped$'':
\begin{equation}
q_i^{\mathrm{tipped}*}\left(T;\bar{\mathcal{S}}\right)  = \frac{\left(1+\tau\right)\alpha_{i}+\lambda \frac{dV^{\mathrm{tipped}}\left(T;\bar{\mathcal{S}}\right)}{dT}}{\left(1+\tau\right)\beta_{i}} 
\end{equation}

\begin{equation}
rV^{\mathrm{tipped}} = \left(1+\tau\right)\sum_{i\in \bar{\mathcal{S}}} \left(\frac{\beta_{i}}{2}(q_{i}^{\mathrm{tipped}*})^2 \right) - (1+\ell)\sum_{i\in \bar{\mathcal{S}}} \left(\frac{1}{2}\rho_{i}T^{2}+\eta_{i}T\right)
\end{equation}

Taking the difference between the State ``$pretip$'' and State ``$tipped$'' HJB equations yields:
\begin{equation}
\begin{aligned}
(r + H(T))&\left(V^{\mathrm{pretip}}(T;\bar{\mathcal{S}}) - V^{\mathrm{tipped}}(T;\bar{\mathcal{S}})\right) +\sum_{i\in \bar{\mathcal{S}}} \left( (\frac{1}{2}\rho_{i}T^{2}+\eta_{i}T) - (1+\ell)(\frac{1}{2}\rho_{i}T^{2}+\eta_{i}T) \right) \\
& = \frac{1}{2} \sum_{i\in \bar{\mathcal{S}}} (1+\tau)\beta_i \left( q_{i}^{\mathrm{pretip}*} - q_{i}^{\mathrm{tipped}*} \right)\left( q_{i}^{\mathrm{pretip}*} + q_{i}^{\mathrm{tipped}*} \right)\\
& = \frac{1}{2} \lambda \left( \frac{dV^{\mathrm{pretip}}\left(T;\bar{\mathcal{S}}\right)}{dT} - \frac{dV^{\mathrm{tipped}}\left(T;\bar{\mathcal{S}}\right)}{dT} \right) \sum_{i\in \bar{\mathcal{S}}} \left( q_{i}^{\mathrm{pretip}*} + q_{i}^{\mathrm{tipped}*} \right) \\
& = \frac{1}{2}(1+\tau)\beta_i(q_i^{\mathrm{pretip}*} - q_i^{\mathrm{tipped}*})\sum_{i^{'}\in \bar{\mathcal{S}}} \left( q_{i^{'}}^{\mathrm{pretip}*} + q_{i^{'}}^{\mathrm{tipped}*} \right)
\end{aligned}
\end{equation}

Solving for the emission difference gives:
\begin{equation}
q_i^{\mathrm{pretip}*} - q_i^{\mathrm{tipped}*} = \frac{2 \left[ (r + H(T))\left(V^{\mathrm{pretip}}\left(T;\bar{\mathcal{S}}\right) - V^{\mathrm{tipped}}\left(T;\bar{\mathcal{S}}\right)\right) - \ell\sum_{i\in \bar{\mathcal{S}}} (\frac{1}{2}\rho_{i}T^{2}+\eta_{i}T) \right]}{(1+\tau)\beta_i \sum_{j\in \bar{\mathcal{S}}} \left( q_{j}^{\mathrm{pretip}*} + q_{j}^{\mathrm{tipped}*} \right)}
\end{equation}

\textbf{Emission differences between State ``$pretip$'' and State ``$notip$'' at the threshold temperature}: The corresponding HJB equation of State ``$notip$'' is:
\begin{equation}
rV^{\mathrm{notip}} = (1+\tau)\sum_{i\in \bar{\mathcal{S}}} \left(\frac{\beta_{i}}{2}(q_{i}^{\mathrm{notip}*})^2 \right)- \sum_{i\in \bar{\mathcal{S}}} \left(\frac{1}{2}\rho_{i}T^{2}+\eta_{i}T\right)
\end{equation}

Taking the first-order condition yields the optimal emissions strategy:
\begin{equation}
q_{i}^{\mathrm{notip}*}\left(T;\bar{\mathcal{S}}\right)=\frac{\left(1+\tau\right)\alpha_{i}+\lambda \frac{dV^{\mathrm{notip}}\left(T;\bar{\mathcal{S}}\right)}{dT}}{\left(1+\tau\right)\beta_{i}}
\end{equation}

Based on this equation, we establish the relationship between the State ``$pretip$'' and State ``$notip$'':
\begin{equation}
\begin{aligned}
&H(T)\left(V^{\mathrm{pretip}}(T;\bar{\mathcal{S}})  - V^{\mathrm{tipped}}(T;\bar{\mathcal{S}})\right) + r\left(V^{\mathrm{pretip}}(T;\bar{\mathcal{S}})  - V^{\mathrm{notip}}(T;\bar{\mathcal{S}})\right)  \\
& = \sum_{i\in \bar{\mathcal{S}}} \frac{\left(1+\tau\right)\beta_{i}}{2}(q_i^{\mathrm{pretip*}})^2 - \sum_{i\in \bar{\mathcal{S}}} \frac{\left(1+\tau\right)\beta_{i}}{2}(q_i^{\mathrm{notip*}})^2\\
& = \frac{1}{2} \sum_{i\in \bar{\mathcal{S}}} (1+\tau)\beta_i \left( q_{i}^{\mathrm{pretip}*} - q_{i}^{\mathrm{notip}*} \right)\left( q_{i}^{\mathrm{pretip}*} + q_{i}^{\mathrm{notip}*} \right)\\
& = \frac{1}{2} \lambda \left( \frac{dV^{\mathrm{pretip}}\left(T;\bar{\mathcal{S}}\right)}{dT} - \frac{dV^{\mathrm{notip}}\left(T;\bar{\mathcal{S}}\right)}{dT} \right) \sum_{i\in \bar{\mathcal{S}}} \left( q_{i}^{\mathrm{pretip}*} + q_{i}^{\mathrm{notip}*} \right)\\
& = \frac{1}{2}(1+\tau)\beta_i(q_i^{\mathrm{pretip}*} - q_i^{\mathrm{notip}*})\sum_{i^{'}\in \bar{\mathcal{S}}} \left( q_{i^{'}}^{\mathrm{pretip}*} + q_{i^{'}}^{\mathrm{notip}*} \right)
\end{aligned}
\end{equation}

Taking the limit as the temperature approaches the tipping threshold ($T \to \underline{T}^+$):
\begin{equation}
\begin{aligned}
& \lim_{T \to \underline{T}^+} \left( q_i^{\mathrm{pretip}*}(T;\bar{\mathcal{S}}) - q_i^{\mathrm{notip}*}(T;\bar{\mathcal{S}}) \right) \\
=& \lim_{T \to \underline{T}^+}\frac{2 \left[ H(T)\left(V^{\mathrm{pretip}}(T;\bar{\mathcal{S}})  - V^{\mathrm{tipped}}(T;\bar{\mathcal{S}})\right) + r\left(V^{\mathrm{pretip}}(T;\bar{\mathcal{S}})  - V^{\mathrm{notip}}(T;\bar{\mathcal{S}})\right) \right]}{(1+\tau)\beta_i \sum_{i^{'}\in \bar{\mathcal{S}}} \left( q_{i^{'}}^{\mathrm{pretip}*}(T;\bar{\mathcal{S}}) + q_{i^{'}}^{\mathrm{notip}*}(T;\bar{\mathcal{S}}) \right)} \\
=& \lim_{T \to \underline{T}^+}\frac{2r\left(V^{\mathrm{pretip}}(T;\bar{\mathcal{S}})  - V^{\mathrm{notip}}(T;\bar{\mathcal{S}})\right) }{(1+\tau)\beta_i \sum_{i^{'}\in \bar{\mathcal{S}}} \left( q_{i^{'}}^{\mathrm{pretip}*}(T;\bar{\mathcal{S}}) + q_{i^{'}}^{\mathrm{notip}*}(T;\bar{\mathcal{S}}) \right)} \\
=& \frac{2r\left(V^{\mathrm{pretip}}(\underline{T};\bar{\mathcal{S}})  - V^{\mathrm{notip}}(\underline{T};\bar{\mathcal{S}})\right) }{(1+\tau)\beta_i \sum_{i^{'}\in \bar{\mathcal{S}}} \left( q_{i^{'}}^{\mathrm{pretip}*}(\underline{T};\bar{\mathcal{S}}) + q_{i^{'}}^{\mathrm{notip}*}(\underline{T};\bar{\mathcal{S}}) \right)}
\end{aligned}
\end{equation}

Building upon these limiting conditions, since $V^{\mathrm{pretip}}(\underline{T};\bar{\mathcal{S}}) < V^{\mathrm{notip}}(\underline{T};\bar{\mathcal{S}})$, the mere exposure to a tipping threat strictly induces a downward shift in the optimal emission path. Consequently, as the temperature reaches the threshold, the grand coalition abruptly reduces its emissions, transitioning into the pre-tipping regime.

However, the emission dynamics immediately before and after the actual occurrence of the tipping event exhibit considerable complexity. While it is structurally guaranteed that $V^{\mathrm{pretip}}(T;\bar{\mathcal{S}}) > V^{\mathrm{tipped}}(T;\bar{\mathcal{S}})$, whether the emission level increases or decreases post-tipping depends entirely on the net balance between $(r + H(T))\left(V^{\mathrm{pretip}}\left(T;\bar{\mathcal{S}}\right) - V^{\mathrm{tipped}}\left(T;\bar{\mathcal{S}}\right)\right)$, and the tipping damage, $\ell\sum_{i\in \bar{\mathcal{S}}} D_i(T(t))$. This trade-off yields three distinct behavioral responses in the post-tipping regime:

Emission Decrease: If ($(r + H(T))\Delta V > \ell\sum_{i\in \bar{\mathcal{S}}} D_i(T(t))$), then $q_i^{\mathrm{pretip}*} > q_i^{\mathrm{tipped}*}$. The occurrence of tipping event forces the coalition's emissions to drop suddenly.

Emission Increase: If ($(r + H(T))\Delta V < \ell\sum_{i\in \bar{\mathcal{S}}} D_i(T(t))$), then $q_i^{\mathrm{pretip}*} < q_i^{\mathrm{tipped}*}$. The occurrence of tipping event acts as a release of intense pre-tipping pressure, leading to a paradoxical increase in emissions.

Continuous Path: If the two forces perfectly balance, the emission path remains continuous, with no immediate jump across the tipping boundary.

\subsection{Proof of Theorem \ref{theorem.3}}
\label{appB3}

This appendix examines how the connection rate $\tau$ and the tipping damage rate $\ell$ affect emissions under full cooperation after tipping has occurred. In the post-tipping regime, damages are given by $(1+\ell)D_i(T)$. Since the HJB equations in the no-tipping regime are identical to those in the post-tipping regime when tipping damage is zero (i.e., $\ell=0$), the following discussion about the impact of the connection rate in the post-tipping regime can also be applied to the no-tipping regime. 

Consider heterogeneous regions under full cooperation, so that all regions belong to the grand coalition. The aggregate value function is quadratic in temperature:
\begin{equation}
V^F(T)=v_0^F+v_1^FT+\frac{1}{2}v_2^FT^2.
\end{equation}
Define
\begin{equation}
A\equiv\sum_{k=1}^n\frac{\alpha_k}{\beta_k},
\qquad
G\equiv\sum_{k=1}^n\frac{1}{\beta_k},
\qquad
R\equiv\sum_{k=1}^n\rho_k,
\qquad
E\equiv\sum_{k=1}^n\eta_k.
\end{equation}

For an interior solution, the grand-coalition HJB equation is
\begin{equation}
\begin{aligned}
rV^F(T)
=
\max_{\{q_i\}_{i=1}^n}
\Bigg\{
&\sum_{i=1}^n
\left[
(1+\tau)
\left(
\alpha_iq_i-\frac{1}{2}\beta_iq_i^2
\right)
-(1+\ell)D_i(T)
\right] \\
&+\lambda V_T^F(T)\sum_{i=1}^nq_i
\Bigg\}.
\end{aligned}
\end{equation}
The first-order condition for region $i$ is
\begin{equation}
(1+\tau)(\alpha_i-\beta_iq_i^F)
+\lambda V_T^F(T)=0.
\end{equation}
Hence,
\begin{equation}
q_i^F(T,\tau,\ell)
=
\frac{\alpha_i}{\beta_i}
+
\frac{\lambda}{(1+\tau)\beta_i}
\left(v_1^F+v_2^FT\right).
\end{equation}

Substituting the optimal emissions into the HJB equation and matching the quadratic terms in $T$ gives
\begin{equation}
\frac{1}{2}rv_2^F
=
\frac{\lambda^2(v_2^F)^2}{2(1+\tau)}G
-\frac{1}{2}(1+\ell)R.
\end{equation}
The stabilizing solution is
\begin{equation}
v_2^F
=
\frac{1+\tau}{2\lambda^2G}
\left[
r-\Omega_F^H(\tau,\ell)
\right],
\end{equation}
where
\begin{equation}
\Omega_F^H(\tau,\ell)
\equiv
\sqrt{
r^2+
\frac{4\lambda^2G(1+\ell)R}{1+\tau}
}.
\end{equation}
Since $R>0$, $1+\tau>0$, and $1+\ell>0$,
$\Omega_F^H(\tau,\ell)>r$, implying $v_2^F<0$.

Matching the linear terms gives
\begin{equation}
rv_1^F
=
\frac{\lambda^2v_1^Fv_2^F}{1+\tau}G
-(1+\ell)E
+\lambda Av_2^F.
\end{equation}
Using
\begin{equation}
\frac{\lambda^2Gv_2^F}{1+\tau}
=
\frac{r-\Omega_F^H(\tau,\ell)}{2},
\end{equation}
we obtain
\begin{equation}
v_1^F
=
\frac{
\frac{(1+\tau)A}{\lambda G}
\left[
r-\Omega_F^H(\tau,\ell)
\right]
-
2(1+\ell)E
}{
r+\Omega_F^H(\tau,\ell)
}.
\end{equation}

Substituting $v_1^F$ and $v_2^F$ into the optimal emission rule yields
\begin{equation}
\label{eq.q.heterogeneous.full.connection.damage}
\begin{aligned}
q_i^F(T,\tau,\ell)
=
&\frac{\alpha_i}{\beta_i}
+
\frac{A}{G\beta_i}
\frac{
r-\Omega_F^H(\tau,\ell)
}{
r+\Omega_F^H(\tau,\ell)
}\\
&-
\frac{
2\lambda(1+\ell)E
}{
(1+\tau)\beta_i
\left[
r+\Omega_F^H(\tau,\ell)
\right]
}
+
\frac{
r-\Omega_F^H(\tau,\ell)
}{
2\lambda G\beta_i
}T.
\end{aligned}
\end{equation}

We first consider the effect of the connection rate. For notational simplicity, let
$\Omega\equiv\Omega_F^H(\tau,\ell)$. Then
\begin{equation}
\frac{\partial\Omega}{\partial\tau}
=
-\frac{
2\lambda^2G(1+\ell)R
}{
(1+\tau)^2\Omega
}
<0.
\end{equation}
Moreover,
\begin{equation}
r+\Omega
+
(1+\tau)\frac{\partial\Omega}{\partial\tau}
=
\frac{(r+\Omega)^2}{2\Omega}
>0.
\end{equation}
Differentiating Equation
\eqref{eq.q.heterogeneous.full.connection.damage}
with respect to $\tau$ and simplifying gives
\begin{equation}
\label{eq.dqdtau.full}
\frac{\partial q_i^F(T,\tau,\ell)}{\partial\tau}
=
\frac{
\lambda(1+\ell)
}{
\beta_i(1+\tau)^2\Omega
}
\left[
E+RT+
\frac{4rA\lambda R}{(r+\Omega)^2}
\right].
\end{equation}

By the condition in Theorem \ref{theorem.3},
\begin{equation}
E+RT
=
\sum_{k=1}^n
\left(\rho_kT+\eta_k\right)
=
\sum_{k=1}^nD_k'(T)
\geq0.
\end{equation}
In addition, $A>0$, $R>0$, $r>0$, $\lambda>0$, $\beta_i>0$,
$1+\tau>0$, and $1+\ell>0$. Therefore, the term in brackets in
Equation \eqref{eq.dqdtau.full} is strictly positive, implying
\begin{equation}
\frac{\partial q_i^F(T,\tau,\ell)}{\partial\tau}>0.
\end{equation}
Thus, under heterogeneous full cooperation, a higher connection rate increases each member's optimal emissions.

We next consider the effect of the tipping damage rate. Since
\begin{equation}
\frac{\partial\Omega}{\partial\ell}
=
\frac{
2\lambda^2GR
}{
(1+\tau)\Omega
}
>0,
\end{equation}
and
\begin{equation}
r+\Omega
-
(1+\ell)\frac{\partial\Omega}{\partial\ell}
=
\frac{(r+\Omega)^2}{2\Omega}
>0,
\end{equation}
differentiating Equation
\eqref{eq.q.heterogeneous.full.connection.damage}
with respect to $\ell$ gives
\begin{equation}
\label{eq.dqdell.full}
\frac{\partial q_i^F(T,\tau,\ell)}{\partial\ell}
=
-\frac{
\lambda
}{
\beta_i(1+\tau)\Omega
}
\left[
E+RT+
\frac{4rA\lambda R}{(r+\Omega)^2}
\right].
\end{equation}
The same condition implies that the term in brackets is strictly positive. Hence,
\begin{equation}
\frac{\partial q_i^F(T,\tau,\ell)}{\partial\ell}<0.
\end{equation}
Thus, under heterogeneous full cooperation, a higher tipping damage rate reduces each member's optimal emissions.

Therefore,
\begin{equation}
\frac{\partial q_i^F(T,\tau,\ell)}{\partial\tau}>0,
\qquad
\frac{\partial q_i^F(T,\tau,\ell)}{\partial\ell}<0,
\qquad \forall i,
\end{equation}
which proves Theorem \ref{theorem.3}.

\section{Appendix C: Parameters and Regions Division}
\label{appC}

\global\long\def\thefigure{C\arabic{figure}}%
\setcounter{figure}{0} 
\global\long\def\thetable{C\arabic{table}}%
\setcounter{table}{0} 
\global\long\def\theequation{C\arabic{equation}}%
\setcounter{equation}{0}

\subsection{Parameters Estimation Results}\label{appC1}

\begin{table}[h!]
\small
\centering
\QTR{caption}{Estimated values of parameters in numerical simulations}
\begin{tabular}{ccccc}
\hline
 & $\alpha$ & $\beta$ & $\rho$ & $\eta$\\
\hline 
China & 11.03 & 3.60 & 0.662 & -0.412\\
US & 24.41 & 14.36 & 0.254 & -0.152\\
EU & 38.18 & 35.01 & 0.331 & -0.178\\
Japan & 27.56 & 72.75 & 0.043 & -0.020\\
Russia & 12.75 & 29.45 & 0.021 & -0.011\\
India & 16.00 & 17.26 & 0.452 & -0.299\\
MidEast & 19.05 & 20.41 & 0.380 & -0.245\\
LatAm & 20.12 & 19.53 & 0.211 & -0.137\\
OthAsia & 11.45 & 9.39 & 0.446 & -0.332\\
Eurasia & 12.19 & 48.25 & 0.026 & -0.017\\
OHI & 22.73 & 39.12 & 0.107 & -0.054\\
Africa & 17.11 & 33.44 & 0.593 & -0.454\\
\hline
\end{tabular}
\centering
\end{table}

Notably, the benefit function parameters $(\alpha_i,\beta_i)$ are calibrated through a three-step procedure. First, using each region's output and emissions in the reference year 2025, we impose a calibration condition that links $\alpha_i$ to $\beta_i$. Second, we estimate $\beta_i$ by nonlinear least squares based on the model's first-order condition evaluated at the 2005, 2015, and 2025 RICE observations. This estimation incorporates the region-specific climate-damage parameters, which are separately calibrated using RICE temperature and damage data over 2005--2105. Finally, given the estimated $\beta_i$, we recover $\alpha_i$ algebraically from the 2025 reference-year calibration condition. This procedure jointly aligns the emissions-associated benefit functions with the regional output, emissions, and climate-damage patterns represented in RICE. 

\begin{figure}[h!]
\includegraphics[width=0.5\linewidth, keepaspectratio]{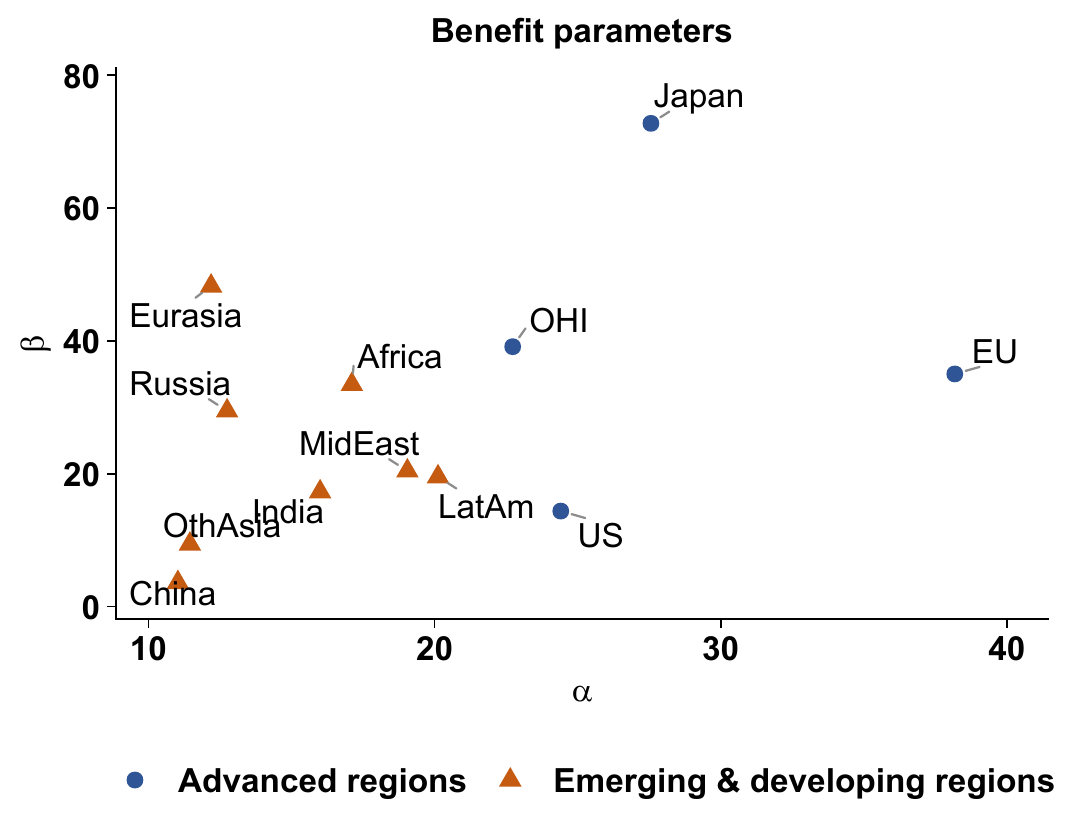}
\QTR{caption}{Regional heterogeneity in calibrated benefit parameters \label{Fig.AppC}}
\begin{figurenotes}
The horizontal axis reports the calibrated parameter $\alpha_i$, while the vertical axis reports $\beta_i$ in the regional benefit function
$\pi_i(q_i)=\alpha_i q_i-1/2\beta_i q_i^2$.
\end{figurenotes}
\end{figure}

Figure \ref{Fig.AppC} illustrates substantial regional heterogeneity in the calibrated benefit parameters. Advanced regions generally exhibit higher values of $\alpha$, indicating relatively larger marginal benefits from emissions at low emission levels. By contrast, emerging \& developing regions tend to have lower values of $\alpha$, while displaying greater heterogeneity in $\beta$. Several of these regions, such as China and OthAsia, have relatively low values of $\beta$, implying a slower decline in marginal benefits as emissions increase, whereas Russia and Eurasia exhibit considerably higher values. Taken together, these estimates imply distinct emission-benefit profiles across regions and therefore contribute to heterogeneity in optimal emission and coalition incentives.

We re‑estimate the parameters using quartic benefit and damage specifications in the sensitivity analyses of Section \ref{sec.sensitivity}. The regional benefit-function parameters display patterns similar to those obtained under the benchmark specification.

\begin{table}[h!]
\small
\centering
\QTR{caption}{Re-estimated parameter values under quartic benefit and damage specifications}

\begin{minipage}{0.48\textwidth}
\centering
\textit{Panel A: Quartic benefit specification}

\vspace{0.2cm}

\begin{tabular}{ccccc}
\hline
 & $\alpha$ & $\beta$ & $\rho$ & $\eta$\\
\hline
China   & 7.06  & 0.22   & 0.662 & -0.412\\
US      & 16.22 & 3.24   & 0.254 & -0.152\\
EU      & 25.23 & 18.93  & 0.331 & -0.178\\
Japan   & 18.36 & 336.97 & 0.043 & -0.020\\
Russia  & 8.49  & 104.42 & 0.021 & -0.011\\
India   & 10.14 & 10.95  & 0.452 & -0.299\\
MidEast & 11.44 & 10.70  & 0.380 & -0.245\\
LatAm   & 13.41 & 12.23  & 0.211 & -0.137\\
OthAsia & 7.47  & 3.86   & 0.446 & -0.332\\
Eurasia & 8.07  & 490.47 & 0.026 & -0.017\\
OHI     & 15.04 & 74.40  & 0.107 & -0.054\\
Africa  & 11.34 & 82.50  & 0.593 & -0.454\\
\hline
\end{tabular}
\end{minipage}
\hfill
\begin{minipage}{0.48\textwidth}
\centering
\textit{Panel B: Quartic damage specification}

\vspace{0.2cm}

\begin{tabular}{ccccc}
\hline
 & $\alpha$ & $\beta$ & $\rho$ & $\eta$\\
\hline
China   & 11.02 & 3.59  & 0.064 & 0.093\\
US      & 24.41 & 14.35 & 0.024 & 0.044\\
EU      & 38.18 & 35.01 & 0.031 & 0.079\\
Japan   & 27.56 & 72.75 & 0.004 & 0.014\\
Russia  & 12.75 & 29.45 & 0.002 & 0.005\\
India   & 15.99 & 17.24 & 0.044 & 0.047\\
MidEast & 19.04 & 20.39 & 0.037 & 0.045\\
LatAm   & 20.12 & 19.53 & 0.020 & 0.025\\
OthAsia & 11.44 & 9.38  & 0.044 & 0.004\\
Eurasia & 12.19 & 48.25 & 0.003 & 0.002\\
OHI     & 22.73 & 39.11 & 0.010 & 0.029\\
Africa  & 17.10 & 33.41 & 0.058 & -0.009\\
\hline
\end{tabular}
\end{minipage}

\end{table}

\subsection{Geographic Division of Regions}\label{appC2}

\begin{table}[h!]
\small
\centering
\caption{Geographic Division of Regions}
\begin{tabularx}{\textwidth}{lX}
\hline
Region & Countries \\
\hline
China & China \\
US & American Samoa, Guam, United States \\
EU & Austria, Belgium, Czech Republic, Denmark, Faeroe Islands, Finland, France, Germany, Greece, Greenland, Hungary, Iceland, Ireland, Italy, Luxembourg, Malta, Netherlands, Norway, Poland, Portugal, Slovak Republic, Spain, Sweden, Switzerland, Turkey, United Kingdom \\
Japan & Japan \\
Russia & Russian Federation \\
India & India \\
MidEast & Bahrain, Cyprus, Iran, Islamic Rep., Iraq, Israel, Jordan, Kuwait, Lebanon, Oman, Qatar, Saudi Arabia, Syrian Arab Republic, United Arab Emirates, West Bank and Gaza, Yemen \\
LatAm & Antigua and Barbuda, Argentina, Aruba, Bahamas, Barbados, Belize, Bermuda, Bolivia, Brazil, Cayman Islands, Chile, Colombia, Costa Rica, Cuba, Dominica, Dominican Republic, Ecuador, El Salvador, Grenada, Guatemala, Guyana, Haiti, Honduras, Jamaica, Mexico, Caribbean Netherlands, Nicaragua, Panama, Paraguay, Peru, Puerto Rico, St. Kitts and Nevis, St. Lucia, St. Vincent and the Grenadines, Suriname, Trinidad and Tobago, Uruguay, Venezuela, U.S. Virgin Islands \\
OthAsia & Afghanistan, Bangladesh, Bhutan, Brunei Darussalam, Cambodia, Fiji, French Polynesia, Indonesia, Kiribati, North Korea, Lao PDR, Malaysia, Maldives, Mongolia, Myanmar, Nepal, New Caledonia, Pakistan, Papua New Guinea, Philippines, Samoa, Solomon Islands, Sri Lanka, Thailand, Timor-Leste, Tonga, Vanuatu, Vietnam \\
Eurasia & Albania, Armenia, Azerbaijan, Belarus, Bosnia and Herzegovina, Bulgaria, Croatia, Estonia, Georgia, Kazakhstan, Kyrgyz Republic, Latvia, Lithuania, North Macedonia, Moldova, Montenegro, Romania, Serbia, Slovenia, Tajikistan, Turkmenistan, Ukraine, Uzbekistan \\
OHI & Australia, Canada, Hong Kong (China), South Korea, Macao (China), New Zealand, Singapore \\
Africa & Algeria, Angola, Benin, Botswana, Burkina Faso, Burundi, Cameroon, Cape Verde, Central African Republic, Chad, Comoros, Democratic Republic of the Congo, Republic of the Congo, Ivory Coast, Djibouti, Egypt, Equatorial Guinea, Eritrea, Ethiopia, Gabon, Gambia, Ghana, Guinea, Guinea-Bissau, Kenya, Lesotho, Liberia, Libya, Madagascar, Malawi, Mali, Mauritania, Mauritius, Morocco, Mozambique, Namibia, Niger, Nigeria, Rwanda, Sao Tome and Principe, Senegal, Seychelles, Sierra Leone, Somalia, South Africa, Sudan, Eswatini, Tanzania, Togo, Tunisia, Uganda, Zambia, Zimbabwe \\
\hline
\end{tabularx}
\end{table}

\end{document}